\documentclass{llncs}[9pt]

\usepackage{comment}
\usepackage{amssymb}
\usepackage{amstext}
\usepackage{amsmath}
\usepackage{latexsym}
\usepackage{tikz}
\usetikzlibrary{trees}
\usetikzlibrary{shapes.geometric, arrows}
\usepackage{qtree}
\usepackage{ifsym}
\usepackage{verbatim}
\usepackage{url}
\usepackage{proof,bussproofs}
\usepackage{listings}

\newcommand{\C}[1]{\mbox{\lstinline`#1`}}

\definecolor{dkblue}{rgb}{0,0.1,0.5} 
\definecolor{lightblue}{rgb}{0,0.5,0.5} 
\definecolor{dkgreen}{rgb}{0,0.4,0} 
\definecolor{dk2green}{rgb}{0.4,0,0} 
\definecolor{dkviolet}{rgb}{0.6,0,0.8}
\definecolor{shadethmcolor}{rgb}{0.9, 0.9,1}

\usepackage{color}
\usepackage{fancyvrb}
\usepackage{paralist}
\usepackage{todonotes}
\newcommand{\ynote}[1]{\todo[inline, color=green!40]{#1}}
\newcommand{\knote}[1]{\todo[inline, color=blue!20]{#1}}

\usepackage{relsize}
\usepackage{float}
\floatstyle{boxed}
\newfloat{boxedfg}{thp}{lop}[figure]
\floatname{boxedfg}{Figure}

\usepackage{enumitem}

\usepackage{mathptmx}
\usepackage[all]{xy}
\usepackage{comment}
\usepackage{multirow}

\newcommand{\mybigcap}{\ensuremath{\mathlarger{\mathlarger{\mathlarger{\cap}}}}}

\newcommand{\mybigcup}{\ensuremath{\mathlarger{\mathlarger{\mathlarger{\cup}}}}}

\newcommand{\Nat}{\mathbb{N}}

\begin{document}
\title{Coinductive Uniform Proofs\vspace*{-0.2in}}
\author{E. Komendantskaya\inst{1} and Y. Li\inst{1} \institute{Heriot-Watt University, Edinburgh,
  Scotland, UK}} 

\maketitle

\begin{abstract}
  Coinduction occurs in two guises in Horn clause logic: 
	in proofs of circular properties and relations, and in proofs involving construction of infinite data.
	Both instances of coinductive reasoning appeared in the literature before, but a systematic analysis of 
	these two kinds of proofs and of their relation was lacking. 
  We propose  a  general proof-theoretic framework for  
	handling both kinds of coinduction arising in Horn clause logic.
	To this aim, we propose a coinductive extension of Miller et al's framework of \emph{uniform proofs}  and prove its soundness
	relative to coinductive models of Horn clause logic.

        This paper is a preliminary version of~\cite{BKL19} and contains a simpler version of the soundness proof that appeared later in~\cite{BKL19}.
        %

\end{abstract}

\section{Introduction}\label{sec:intro}

 Horn clause logic is a fragment of predicate logic, in which all formulae are written in clausal form. For example, the two clauses below inductively define a list membership property: 

\vspace*{0.05in}

\noindent
(1) $\forall x\ y\ (member\ x \ [x |y])$\\ 
(2) $\forall x\ y\ t \ (member\ x\ t \;\supset\; member \ x\ [y | t])$

\vspace*{0.05in}

\noindent We use notation $[\_ | \_]$ for list and stream constructors throughout this paper.
Generally,  a Horn clause has a shape $\forall x_1 \ldots x_m\, (A_1 \land \ldots  \land A_n \supset A)$, where each $A_i$ and $A$
are atomic formulae.
The clause (1) above should be read as having an empty antecedent. 
This syntax implicitly restricts the syntax of a \emph{goal clause (or goal)}, i.e. a clause with empty consequent.
Firstly, all goals are existentially quantified.
Secondly, they may only contain conjunctions.  
 Thus we can have as a goal 
e.g. $\exists x\ (member \ 0 \ [0 | 1 | x] \land member \ 1 \ [0 | 1 | x])$, and prove it with $x = nil$, 
but not an implicative or a universally quantified formula, e.g. not $\forall x\ (member \ 0 \ [0 | 1 | x])$ or
$\forall x\ (member \ 0 \ [0 | 1 | x] \supset member \ 1 \ [0 | 1 | x])$.

Initially introduced by Horn in the 30s,  Horn clause logic gained popularity for 
yielding efficient proofs by resolution~\cite{EO93}, which made it the logic of  choice for  first implementations of Prolog in 70s and 80s~\cite{Llo87}, as well as several resolution-based provers in the 90s.
The next big step was made in the 90s by Miller, Nadathur et. al~\cite{MNPS91,MN12} who studied proof-theoretic properties of Horn clause logic, introducing the notion of a \emph{uniform proof}. 
Within that framework, three  extensions to Horn clause logic were suggested. 
Taking first-order Horn clauses (\emph{fohc}) as a starting point, 
they extended it with
lambda abstraction, obtaining \emph{higher-order Horn clauses} (\emph{hohc}). 
Allowing not only conjunctions, but also implications, disjunctions and quantified formulae in the antecedent of $\supset$ 
gave \emph{hereditary Harrop} logic, which, too can have first-order or higher order syntax resulting in \emph{first-order hereditary Harrop clause} logic (\emph{fohh}) 
and  \emph{higher-order hereditary Harrop clause} logic (\emph{hohh}). 
For example, we can express and prove  
$\forall x\ (member \ 0 \ [0 | 1 | x])$ and
$\forall x\ (member \ 0 \ [0 | 1 | x] \supset member \ 1 \ [0 | 1 | x])$ in \emph{fohh}, taking clauses (0) and (1) as axioms.
Figure~\ref{fig:diamonds} shows the ``uniform proof diamond" (the arrows show syntactic extensions).

\begin{figure}[t]
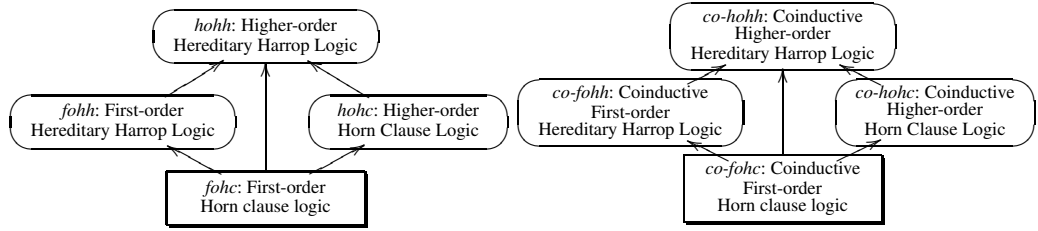

\scriptsize{
$$
\xy0;/r.03pc/: 
(0,0)*[o]=<85pt,20pt>\hbox{\txt{\emph{hohh}: Higher-order\\ Hereditary Harrop  Logic}}="d"*\frm<8pt>{-},
(-150,-90)*[o]=<85pt,20pt>\hbox{\txt{\emph{fohh}: First-order\\ Hereditary Harrop  Logic}}="c"*\frm<8pt>{-},
(150,-90)*[o]=<75pt,20pt>\hbox{\txt{\emph{hohc}: Higher-order\\ Horn Clause Logic }}="b"*\frm<8pt>{-},
(0,-170)*[o]=<75pt,20pt>\hbox{\txt{\emph{fohc}: First-order\\ Horn clause logic }}="a"*\frm<8pt>{-,},
"c";"d" **\dir{-} ?>*\dir{>},
"b";"d" **\dir{-} ?>*\dir{>},
"a";"d" **\dir{-} ?>*\dir{>},
"a";"c" **\dir{-} ?>*\dir{>},
"a";"b" **\dir{-} ?>*\dir{>}
\endxy
\ \ \ 
\xy0;/r.03pc/: 
(0,0)*[o]=<85pt,25pt>\hbox{\txt{\emph{co-hohh}: Coinductive \\ Higher-order\\ Hereditary Harrop  Logic}}="d"*\frm<8pt>{-},
(-160,-80)*[o]=<85pt,25pt>\hbox{\txt{\emph{co-fohh}: Coinductive \\ First-order\\ Hereditary Harrop  Logic}}="c"*\frm<8pt>{-},
(160,-80)*[o]=<75pt,25pt>\hbox{\txt{\emph{co-hohc}: Coinductive \\ Higher-order\\ Horn Clause Logic }}="b"*\frm<8pt>{-},
(0,-160)*[o]=<75pt,25pt>\hbox{\txt{\emph{co-fohc}: Coinductive \\ First-order\\ Horn clause logic }}="a"*\frm<8pt>{-,},
"c";"d" **\dir{-} ?>*\dir{>},
"b";"d" **\dir{-} ?>*\dir{>},
"a";"d" **\dir{-} ?>*\dir{>},
"a";"c" **\dir{-} ?>*\dir{>},
"a";"b" **\dir{-} ?>*\dir{>}
\endxy
$$}
\caption{\footnotesize{\textbf{Left: uniform proof diamond by Miller et al.} 
\textbf{Right: coinductive uniform proof diamond proposed in this paper}}}\label{fig:diamonds}
\end{figure}

To name a few most recent appearances of these logics in the literature, higher-order Horn clauses are gaining popularity in verification of functional programs and in refinement types~\cite{HashimotoU15,BurnOR17}; and first-order hereditary Harrop clause logic has been used in extensions of Haskell type classes~\cite{BottuKSOW17}.
There are of course $\lambda$Prolog/Teyjus~\cite{NadathurM99} 
and Abella~\cite{BaeldeCGMNTW14} influenced by uniform proofs, as well as an extensive SMT-solving literature on applications of Horn clauses in verification~\cite{BjornerGMR15}. 
The uniform proof diamond thus gives a useful proof-theoretic classification for a set of computationally tractable and interesting 
fragments of first-order and higher-order logics.

Coinductive properties of first-order Horn clause logic are not equally well understood proof-theoretically. Firstly, we can define coinductive (or infinite) data structures, such as streams or infinite trees, in this logic. For example an infinite stream of bits can be defined as follows:

\vspace*{0.05in}
\noindent 
(3) $\forall x\ y\ ((bitstream\ y \;\land\; bit\ x ) \;\supset\; bitstream [ x | y ] )$\\ 
(4) $bit(0)$ \\
(5) $bit(1)$

\vspace*{0.05in}

\noindent Given a \emph{suitable coinductive proof principle}, we can prove 
$\exists \ y (bitstream\ [0| y])$, substituting $y$ with a suitable recursive term. 
Some such derivations can be performed in CoLP~\cite{GuptaBMSM07,SimonBMG07}.

Secondly, first-order Horn clauses can be used to state and prove circular properties and relations.
These do not have to involve infinite data structures.
We may state the following circular property of the (finite) list membership: 

\vspace*{0.05in}
\noindent 
(6) $\forall x\ y\ t \ ((member\ x\ [y | t] \land eq \ x \ y) \;\supset\; member \ x\ [y | t])$\\
(7) $\forall x\ (eq\ x\ x)$
\vspace*{0.05in}

\noindent Again, given \emph{a suitable coinductive proof principle}, the goal $member(0, [0| nil])$ could be proven coinductively by using just the clauses (6) and (7), without inductively analysing the list structure as clauses (1) and (2) suggest. 
Interesting cases of such circular proofs were reported in implementation of Haskell type classes~\cite{LammelJ05,FKS15}.

Thridly, the existing literature tends to regard coinduction in Horn clause logic mainly operationally, as a process of loop detection. 
The question is thus two-fold: \emph{Can we formulate a coinductive proof principle for Horn clause logic, complementing the work  of Miller et al. on systematisation of proof-theoretic properties of Horn clause logic?} \emph{What would be the ``suitable coinductive proof principle" to cover both cases of coinduction?}
This question has not been raised or answered in the literature in this generality. And this is where this paper fills the gap.  

Suppose  we simply take \emph{fohc} and extend it with a standard coinductive proof rule \emph{\`{a} la}
 Coquand and Gimenez~\cite{Coq94,Gimenez98}.
Would this resulting coinductive calculus, lets call it \emph{co-fohc}, be suitable to prove coinductive properties of \emph{fohc} formulae?
-- It turns out not: we would be able to prove $member(0, [0| nil])$ given the clauses (6)-(7), but not $\exists \ y (bitstream\ [0| y])$ for clauses (3)-(5). 
To instantiate the existential quantifier,  we would need to define the stream of zeros using a fixed-point operator, and for this we need at least the language of \emph{hohc}. 

Next, suppose we take \emph{co-hohc}, i.e. a coinductive version of \emph{hohc}.
Would it give us a suitable coinductive calculus?  -- The answer is again negative. It would cover our earlier examples, but not the following example, defining an infinite stream of successive numbers:

\vspace*{0.05in}
\noindent 
(8) $\forall x\ y\ ( from\ (s\ x)\ y \;\supset\; from\ x\ [ x\ | y ] )$
\vspace*{0.05in}

\noindent Given $\exists y \ (from\ 0\ y)$ as a goal, $[0\ |\ (s\ 0)\ |\ (s\ s\ 0) \ |\ \ldots]$ would be an intended answer for $y$.
For this clause, we cannot prove  $\exists \ y\ (from\ 0\ y)$ directly, but we have to prove a more general lemma first, i.e.
$\forall \ x\ (from\ x\ (fr\_str\ x))$, where $fr\_str\ x$ is a fixed point definition of a stream of successive numbers starting with $x$. 
But, by the earlier discussion, this universal goal does not belong to \emph{hohc}! It is a goal of hereditary Harrop clause logic, that is, we need at least the syntax and rules of \emph{hohh}, with additional coinductive rule (i.e. we need a calculus \emph{co-hohh}).

To complete the picture with a fourth calculus, we consider an example of a coinductive property whose proof falls within the remits of \emph{co-fohh}.
Suppose we are given a definition of a comembership property, i.e. a property that an element occurs in a stream infinitely many times\footnote{It can be defined for example as follows:
 $\forall\ x\ t\ (drop\ x\ [x \ | t]\ t)$\\
 $\forall\ x\ t \ s\ (drop\ x\ t\ s\ \;\supset\; drop\ x \ [\_ \ | \ t] \ s) $\\
 $\forall\ x\ t \ s\ ((drop\ x\ t\ s\ \land comember\ x\ s) \;\supset\; comember\ x \ t)$}.
Suppose now a function $f$ is defined to take streams of bits as input and return streams of bits as an output. 
Therefore, its application would not affect the property of a bit comembership:

\vspace*{0.05in}
\noindent 
(9) $\forall x\ s\ ( (comember_{bit} \ x\ (f s) \land bit \ x ) \;\supset\; comember_{bit} \ x \ s )$
\vspace*{0.05in}

\noindent Again, using a suitable coinductive proof principle, we could  prove the lemma $\forall x\ s\ ( bit\ x \;\supset\; comember_{bit} \ x \ s )$.
This lemma would belong to \emph{co-fohh}, as its syntax is expressed in first-order hereditary Harrop clause logic. A \emph{fohc} version of the lemma, e.g. $comember_{bit} \ 0 \ str$ is not provable by coinduction (whatever $str$ is), but would follow from $\forall x\ s\ ( bit\ x \;\supset\; comember_{bit} \ x \ s )$ if we prove that first. Many more examples of this kind of coinductive reasoning are given in~\cite{FKS15}. 

In Section~\ref{sec:cup}, we introduce these four calculi formally and thereby establish a \emph{uniform proof} framework for systematic study of coinduction that arises in Horn clause logic (cf. Figure~\ref{fig:diamonds}). 
%
%
%
We show that these four calculi characterize four different classes of coinductive properties arising from Horn clause definitions: \emph{co-fohc} and \emph{co-fohh} are required when proofs do not involve construction of infinite data, whereas  
\emph{co-hohc} and \emph{co-hohh} are needed for proofs involving infinite data construction. These cases further sub-divide: \emph{co-fohc} and \emph{co-hohc} are suitable for regular proofs of atomic coinductive properties (like $member(0, [0| nil])$).  \emph{Co-fohh} and \emph{co-hohh} will be needed to construct proofs that exhibit irregularity 
and require more general coinductive hypotheses
(like $\forall \ x\ (from\ x\ (fr\_str\ x))$ or $\forall x\ s\ ( bit\ x \;\supset\; comember_{bit} \ x \ s )$).
As if by magic, Miller et al's  ``uniform proof diamond"
happens to give a perfect basis for systematic description of coinductive proofs and properties 
arising in Horn clause logic. 


Let us call a set of  first-order Horn clauses a \emph{logic program}. 
In Section~\ref{sec:sound} we show that, given a logic program $P$, the proofs obtained for $P$ in \emph{co-fohc, co-fohh, co-hohc, co-hohh} are coinductively sound, i.e. sound relative to the \emph{greatest Herbrand} model of $P$. We also show that it is coinductively sound to use coinductively proven lemmas in other proofs, as was illustrated  in the final example of this section.

First-order Horn clause logic is Turing complete and therefore any sound calculus for it will fail to decide some recursive cases. 
In Section~\ref{sec:rw}, we discuss examples  of logic programs that define coinductive properties
beyond the power of our ``coinductive diamond". 

The uniform proofs in general, and their coinductive version in particular, are not designed to compete in expressivity or power with richer logics, such as e.g. the calculus of (coinductive) constructions~\cite{Coq94,Gimenez98}.  The motivation behind uniform proofs is to give an elegant and thorough proof theoretic analysis of automated reasoning with Horn clauses. 
Taking this perspective, the presented work advances the state of knowledge in several ways: 
 
\begin{itemize}[leftmargin=3mm] 
\item The ``coinductive diamond" covers and extends the currently available coinductive algorithms for Horn clause logic. For example,
Co-LP~\cite{GuptaBMSM07,SimonBMG07} is equivalent to \emph{co-fohc} plus  \emph{co-hohc},  
productive corecursion of \cite{KL17} is equivalent to \emph{co-hohc}, but neither extends to \emph{co-fohh} or \emph{co-hohh}. Methods introduced in ~\cite{LammelJ05,FKS15} are subsumed by \emph{co-fohc} and \emph{co-fohh}, but  cannot work with infinite data, 
thus fall short of \emph{co-hohc} and \emph{co-hohh}. 
\item This paper shows a novel application of the coinductive proof principle known in functional programming~\cite{Coq94,Gimenez98} to uniform proofs. 
We take this principle to a terrain where it has considerably less support than in Coq to ensure its soundness: 
e.g. we cannot rely on (co)-inductive data types or proof terms (on which guardedness checks are usually imposed, cf. \cite{Gimenez98}).
\item Our soundness result for \emph{co-hohh} is conceptually  novel, 
and, when restricted to the other three calculi, addresses limitations of similar results published in the literature:
we replace the operational soundness of  \emph{co-fohh} given in~\cite{FKS15} with soundness relative to coinductive models; 
we replace non-constructive proof of soundness of \emph{co-fohh} given in~\cite{FKH16} with a constructive proof; we extend soundness results of \cite{GuptaBMSM07,SimonBMG07} to \emph{co-hohh} and \emph{co-fohh}.
			\item Overall, the presented results advance our understanding of coinductive properties expressible in first-order logic. 
			We thus hope that this work will give a solid formal support for new implementations of coinduction in automated proving based on first-order logic. 
\end{itemize}
 
The paper proceeds as follows.
Section~\ref{sec:backgr} gives the necessary background definitions. 
Section~\ref{sec:cup} introduces the four original coinductive calculi based on uniform proofs. Section~\ref{sec:sound} proves that 
our coinductive calculi are sound relative to the coinductive model of Horn clause logic.
Finally, Section~\ref{sec:rw} concludes and discusses related and future work.


\section{Preliminaries, Fixed-point Terms}\label{sec:backgr}

We first recall some essential definitions concerning the syntax of uniform proofs~\cite{MN12}. In order to be able to work with 
infinite data structures like streams, we extend that original syntax
with fixed-point combinators in a standard way, following~\cite{Pierce02,Nederpelt}. 

Let $\mathbb{V}=\{o,\iota, \ldots\}$ be a countable set of \emph{type variables}, and $\mathbb{T}$ be the set of all \emph{simple types} (in short, \emph{types}) in the abstract syntax:  \[\mathbb{T} ::= \mathbb{V}\mid \left(\mathbb{T}\rightarrow\mathbb{T}\right)\]
  We use $\tau,\sigma$ and variants 
	thereof to denote an arbitrary type.  A type  $\sigma\in \mathbb{V}$ is called a \emph{non-functional} type. A type of the form $\tau\rightarrow\tau'$ is called a \emph{functional} type. 
 Given the right associativity of the functional type constructor ($\rightarrow$), we can remove all redundant parentheses and depict a type  $\tau$ in the form $\tau_1\rightarrow\cdots\rightarrow\tau_n\rightarrow\sigma\ \left(n\geq 0\right)$, where the non-functional type $\sigma$ is the \emph{target type} of $\tau$, and $\tau_1,\ldots\tau_n$ are \emph{argument types} of $\tau$.

\begin{definition}[Terms]
Let  $\mathit{Var}=\{x,y,z,\ldots\}$ be a countable set of \emph{(term) variables}, and $\mathit{Cst}=\{a,b,c,\ldots\}$ -- a countable set of \emph{constants} .    The set $\Lambda$ of all \emph{lambda terms} (in short: \emph{terms}) is defined by 
\[\Lambda ::= \mathit{Var} \mid \mathit{Cst} \mid \left(\Lambda \  \Lambda\right) \mid \left(\lambda \mathit{Var}\, . \, \Lambda\right) \mid \left(\textit{fix}\ \lambda \mathit{Var}\, . \, \Lambda\right) \]  
A term of the form $\Lambda\ \Lambda$  is  called an \emph{application}, and a term of the form $\lambda \mathit{Var}\, . \, \Lambda$  -- \emph{abstraction}.  
A term of the form $\textit{fix}\ \lambda \mathit{Var}\, . \, \Lambda$ is called a \emph{fixed-point term}.
  We use letters $M, N, L$ and variants thereof to denote members of $\Lambda$.  
\end{definition}
	 Given a term $M$, the set $FV\left(M\right)$ of \emph{free variables}  and 
	the set  $\textit{Sub}\left(M\right)$  of sub-terms of $M$ are defined as standard, cf. Appendix~\ref{sec:syntax}.  
 Term $M$ is \emph{closed}, if $FV\left(M\right)=\emptyset$. Otherwise, $M$ is  \emph{open}. We define $\Sigma$, called a \emph{signature}, be a partial function mapping from $\mathit{Cst}$ to the set $\mathbb{T}$ of simple types, and $\Gamma$, called a \emph{context}, be a partial function mapping from $\mathit{Var}$ to $\mathbb{T}$.
	A term $M$ is \emph{well typed} if it satisfies type judgment $\Sigma,\Gamma \vdash M:\tau$ for some $\Sigma,\Gamma$ and $\tau$, using the standard typing rules given in Appendix~\ref{sec:syntax}. 
	We will only work with well typed terms.
	All \emph{free} variables in a term must have their types assigned by $\Gamma$. Therefore, we use $\Sigma,\emptyset \vdash M:\sigma$ to imply that $M$ is a \emph{closed} term.

We use the symbols $\land$ for \emph{conjunction}, $\lor$ for \emph{disjunction}, $\supset$ for \emph{implication}, and $\top$ for the \emph{true} proposition. Given some type $\tau$, we use $\forall\!_\tau$ for \emph{universal quantification}, and $\exists_\tau$ for \emph{existential quantification}, over terms of type $\tau$. These symbols are special constants, called \emph{logical constants} and their types are as follows:  $\land,\lor,\supset:o \rightarrow o \rightarrow o$, and $\top:o$, and $\forall\!_\tau,\exists_\tau: \left(\tau \rightarrow o\right) \rightarrow o$.  
Quantification over a term $M$ is a short hand for applying the quantifier to an abstraction over $M$.  
We assume that a signature $\Sigma$ always contains at least all logical constants.

\begin{definition}[First-order types and terms]
The \emph{order} of a type $\tau$, denoted $\mathit{Ord}\left(\tau \right)$, is
 \begin{align*}
 \mathit{Ord}\left(\tau\right) & = 0\quad  \text{(provided that $\tau$ is non-functional)}\\
 \mathit{Ord}\left(\tau_1\rightarrow\tau_2\right) & = \max \{\mathit{Ord}\left(\tau_1\right) + 1,\ \mathit{Ord}\left(\tau_2\right)\}    
 \end{align*}
A term $M:\sigma$ is called \emph{first order} if the following conditions are satisfied.
\begin{enumerate}
\item $\mathit{Ord}\left(\sigma\right)= 0$, and
\item $\{\mathit{Ord}\left(\tau\right)\mid {N:\tau}\in\mathit{Sub}\left(M\right) \land N\in\textit{Cst}\}\subseteq\{0,1\}$, and
\item $\{\mathit{Ord}\left(\tau\right)\mid {N:\tau}\in\mathit{Sub}\left(M\right) \land N\in\textit{Var}\}\subseteq\{0\}$, and \label{enum item: variable in first order term }
\item $o\notin\{\tau\mid {N:\tau}\in\mathit{Sub}\left(M\right)\}$, and
\item $\emptyset=\{N\in\mathit{Sub}\left(M\right)\mid N\text{ is a fixed point term}\}$.
\end{enumerate}
\end{definition}

 A \emph{predicate} is  a variable or a non-logical constant of type $\tau$ such that either $\tau=o$ or the target type of $\tau$ is $o$. We say that a predicate is \emph{first order} if it is a non-logical constant, whose type expression is of order at most 1, and the type $o$ does  not  occur as its argument type. Otherwise a predicate is of \emph{higher order}. 
%
%
A signature  $\Sigma$ is  \emph{first-order} if for all non-logical constants $b$ in $\Sigma$, (i) the type expression of $b$ is at most of order $1$, and (ii) the type $o$ is involved in the type expression of  $b$ only if $b$ is a first order predicate.

\emph{Substitution} of all free occurrences of a variable $x$ in a term $M$ by a term $N$, denoted  $M\left[x:=N\right]$, and the notion of $\alpha$-equivalence are defined as standard, see Appendix~\ref{sec:syntax}. 
 We use $\equiv$ to denote the relation of \emph{syntactical identity modulo $\alpha$-equivalence}.
We call $\left[x_1:=N_1\right]\ldots\left[x_m:=N_m\right]$
a \emph{substitution} when we have
$M\left[x_1:=N_1\right]\ldots\left[x_m:=N_m\right]$ (the latter denotes $ \left(\ldots\left(M\left[x_1:=N_1\right]\right)\ldots\right)\left[x_m:=N_m\right] $).

\emph{One step $\beta$-reduction}, denoted $\rightarrow_\beta$, is defined as $\left(\lambda x\,.\, M\right)N\rightarrow_\beta M\left[x:=N\right]$. Moreover, it is defined that if $M\rightarrow_\beta N$, then $M\ L \rightarrow_\beta N\ L$, $L\ M \rightarrow_\beta L\ N$, $\lambda z \, . \, M  \rightarrow_\beta \lambda z \, . \, N$ and $\textit{fix}\ M \rightarrow_\beta \textit{fix}\ N$. If there is no term $N$ such that $M\rightarrow_\beta N$, then $M$ is \emph{$\beta$-normal}. We define  $\beta$-\emph{equivalence}, denoted $=_\beta$, 
as a transitive and reflexive closure of $\rightarrow_\beta$. Because the $fix$ primitive is treated as a constant symbol when doing $\beta$-reduction, it does not affect strong normalisation of $\beta$-reduction. 
In this paper we assume that all terms are in $\beta$-normal form unless otherwise stated.

\emph{One step $\textit{fix}\beta$-reduction}, denoted $\rightarrow_{\textit{fix}\beta}$, is defined as \[\textit{fix}\ \lambda x \ .\ M \rightarrow_{\textit{fix}\beta} M\left[x:= \textit{fix}\ \lambda x \ .\ M\right] \] Moreover, we require that, if $P \rightarrow_{\textit{fix}\beta} N$, then for arbitrary $z$ and $L$, $P\ L \rightarrow_{\textit{fix}\beta} N\ L$, $L\ P \rightarrow_{\textit{fix}\beta} L\ N$, $\lambda z \, . \, P  \rightarrow_{\textit{fix}\beta} \lambda z \, . \, N$ and $\mathit{fix}\ P \rightarrow_{\textit{fix}\beta} \mathit{fix}\ N$. 
${\textit{Fix}\beta}$\emph{-equivalence}, denoted $=_{\textit{fix}\beta}$, is defined as: 
$M=_{\textit{fix}\beta} N$ if there is $n\geq 0$ and there are terms $M_0$ to $M_n$ such that $M_0\equiv M$, $M_n\equiv N$, and for all $i$ such that $0\leq i < n$, $M_i =_\beta M_{i+1}$ or $M_i \rightarrow_{\textit{fix}\beta} M_{i+1}$ or $M_{i+1} \rightarrow_{\textit{fix}\beta} M_{i}$.


\begin{example}[Stream of zeros]\label{eg: zeros as fix point}
Given
non-logical constants  $0:\iota$ and 
$\textit{scons}:\iota \rightarrow \iota \rightarrow \iota$ (\textit{scons} can be identified with $[\_|\_]$),
we represent the stream of zeros in the form of a fixed-point term $\textit{z\_str}=_{\mathit{def}} \textit{fix}\ \lambda x\,.\, \textit{scons}\ \ 0\ \ x  $. The following relations justify that $z\_str$ models the stream of zeros:
\begin{align*}
	\textit{z\_str} & =_{\mathit{def}}  \textit{fix}\ \lambda x\,.\, \textit{scons}\ \ 0\ \ x \\
	&   \rightarrow_{\textit{fix}\beta} \textit{scons}\ \ 0\ \ \textit{z\_str}  & \\
	& \rightarrow_{\textit{fix}\beta} \textit{scons}\ \ 0\ \ \left(\textit{scons}\ \ 0\ \ \textit{z\_str}\right)   &   \\          
	&   \ldots & 
	\end{align*}

\end{example}

\begin{example}[Constant stream]\label{eg: number stream as fix point term}
The stream of zeros can be generalised to a stream of any particular number, 
as follows: $\textit{n\_str}\ =_{\mathit{def}}\ \textit{fix}\ \lambda fn\,.\,\textit{scons}\ n \left(f\ n \right) $. It 
takes a term $x:\iota$ as parameter and returns a stream of $x$'s:
\begin{align*}
\textit{n\_str}\ x & =_{\mathit{def}}  \left(\textit{fix}\ \lambda fn\,.\,\textit{scons}\ n \left(f\ n \right)\right)\ x\\
& =_{\textit{fix}\beta}\  \textit{scons}\ x \left(\textit{n\_str}\ x \right)\\
& =_{\textit{fix}\beta}\  \textit{scons}\ x \left(\textit{scons}\ x \left(\textit{n\_str}\ x \right) \right)    \\
& \ldots                  
\end{align*}

\end{example}

\begin{example}[Stream of successive numbers]\label{eg: from as fix point}
We add a successor function $s:\iota\rightarrow\iota$ to the signature of Example~\ref{eg: zeros as fix point}.
We represent, in the form of a fixed-point term, the family of streams of successive natural numbers: $ \textit{fr\_str}=_{\mathit{def}} \textit{fix}\ \lambda f n\ .\ \textit{scons}\ \ n\  \left(f\  \left(s\ n\right)\right)$.

Again, the following relations hold for all $n:\iota$ (for instance, $n$ can be $0$, $s\; 0$, \ldots ):
\begin{align*}
  \textit{fr\_str}\ n & =_{\mathit{def}} \quad         \left( \textit{fix}\ \lambda f n\ .\ \textit{scons}\ \ n\  \left(f\  \left(s\ n\right)\right)\right) \ n\\
 & =_{\textit{fix}\beta}    \textit{scons}\ \ n\  \left(\textit{fr\_str} \left(s\ n\right) \right)\\
 & =_{\textit{fix}\beta} \    \textit{scons}\ \ n\  \left(\textit{scons}\ \left(s\ n\right)\  \left(\textit{fr\_str} \left(s^2\ n\right) \right)\right)\\
 & =_{\textit{fix}\beta} \    \textit{scons}\ \ n\  \left(\textit{scons}\ \left(s\ n\right)\  \left(\textit{scons}  \left(s^2\ n\right)\  \left(\textit{fr\_str} \left(s^3\ n\right) \right) \right)\right)\\
 &\ldots &
\end{align*} 

\end{example}

In line with the standard literature on this subject~\cite{BK08,BlanchetteM0T17}, we now need to introduce guardedness conditions on fixed-point definitions, otherwise some of them may not define any infinite objects. A detailed discussion of term productivity or various guardedness conditions that insure this property are beyond the scope of this paper. We simply adapt standard guardedness conditions~\cite{Gimenez98}, although other methods available in the literature (e.g.~\cite{BlanchetteM0T17}) may work as well. 

\begin{definition}[Guarded fixed point terms]
	\label{defn: guarded fpt}
A guarded fixed point term has the form \[\mathit{fix}\ \lambda x\,.\,\lambda y_1\ldots y_m\,.\, f\ L_1\ldots L_k\ \left(x\ N_1\ldots N_m\right)\ L_{k+1}\ldots L_r\qquad \textrm{where}\]

\begin{enumerate}[leftmargin=3mm]
	\item $x$ is a variable of type $\tau_1\rightarrow\cdots \rightarrow \tau_m \rightarrow \iota$, where $\tau_1=\ldots =\tau_m =\iota$ and $m\geq 0$.
	\item $f$ is a constant of type $\sigma_1\rightarrow\cdots  \rightarrow \sigma_{r+1} \rightarrow \iota$ where $\sigma_1=\ldots= \sigma_{r+1} =\iota$ and $ r\geq 0$.
	\item $L_1,\ldots, L_r,N_1,\ldots, N_m$ are first order terms of type $\iota$.
	\item $x\notin\{y_1:\iota,\ldots,y_m:\iota\}=\mathit{FV}\left(L_1\right)\cup\cdots\cup\mathit{FV}\left(L_r\right)\cup\mathit{FV}\left(N_1\right)\cup\cdots\cup\mathit{FV}\left(N_m\right)$.
\end{enumerate}
\end{definition}	
Clearly, the generic guarded fixed point term in Definition~\ref{defn: guarded fpt} is closed and has type $\tau_1\rightarrow\cdots \rightarrow \tau_m \rightarrow \iota$. All fixed point terms shown in examples above are guarded.
\begin{definition}[Guarded full terms]
A guarded full term is either a first order term, or a higher order term in the form $F\ M_1\ldots M_n$ where $F:\tau_1\rightarrow\cdots \rightarrow \tau_n \rightarrow \iota$ is a guarded fixed point term, and $M_1:\iota,\ldots, M_n:\iota$ are first order terms, provided $\tau_1=\ldots =\tau_n =\iota$ and $n\geq 0$.  	 
\end{definition}	

An \emph{atomic formula} (in short, \emph{atom}) is a lambda term of type $o$, in the form $h \ M_1\ \ldots\ M_n$, where $n\geq 0$, $h$ is a predicate, and $M_1, \ldots,M_n$ are lambda terms; if $h$ is a variable, then we say the atom is \emph{flexible}; if $h$ is a non-logical constant, then the atom is \emph{rigid}; if $h$ is a first order predicate, and $M_1, \ldots,M_n$ are first order terms, then the atom is of  \emph{first order}.
%
	If  $h$ is a first order predicate, and $M_1, \ldots,M_n$ are guarded full terms, then the atom $A=_{\mathit{def}}h \ M_1\ \ldots\ M_n$ is guarded, and all atoms $A'$ such that $A'=_{\mathit{fix}\beta} A$ are guarded.

\section{Coinductive Uniform Proofs}\label{sec:cup}


We start with introducing 
 the formal languages of four uniform proof calculi by Miller and Nadathur~\cite{MN12}.
As in the original setting, the distinction between Horn clauses and hereditary Harrop clauses lies in enriched syntax for goals: the latter admits universally quantified and implicative goals. The distinction between first-order and higher-order logics  is made by allowing or disallowing higher-order terms, as defined in the previous section.

Assume a signature $\Sigma$. Let $\mathcal{U}^\Sigma_1$ be the set of all terms over $\Sigma$ that do \emph{not} contain the logical constants $\forall$ and $\supset$, and  $\mathcal{U}^\Sigma_2$ be the set of all terms over $\Sigma$ that do \emph{not} contain the logical constant $\supset$ \footnote{Miller and Nadathur called the sets $\mathcal{U}^\Sigma_1$ and $\mathcal{U}^\Sigma_2$  ``Herbrand universes'' \cite[\protect\S 5.2]{MN12}. Here we reserve the name ``Herbrand'' for later modal-theoretic study of the languages.}. 
Further, let $A$ and $A_r$ denote the sets of atoms and rigid atoms, respectively, on $\Sigma$. 
 In the setting of \emph{co-fohc}  and \emph{co-fohh},  let $A$ be first order (we will denote this fact by $A^1$ in Table~\ref{tab:up}). 
In the setting of \emph{co-hohc}, let $A$ and $A_r$ be from $\mathcal{U}^\Sigma_1$; in the setting of  \emph{co-hohh}, let $A$ and $A_r$ be from $\mathcal{U}^\Sigma_2$. 
Using the sets $A^1$, $A_r$, and $A$,  Table~\ref{tab:up} defines, for each of the four calculi, the set $D$ of \emph{program clauses}  and the set $G$ of \emph{goals}.

\begin{table*}[t] 
 \centering
 {\footnotesize{
 \renewcommand{\arraystretch}{1.5}
 \begin{tabular}{p{1.2cm}|| p{4.8cm} ||  p{7.2cm} } 
   \hline
	& Program Clauses & Goals 
	\\ \hline \hline
	\emph{co-fohc} & $D \ ::= A^1 \mid G \supset D \mid  D\land D  \ \mid \forall Var \ D$ &  
	$G \ ::= \top \mid A^1 \mid G \land G \mid  G\lor G  \mid \exists Var \ G$ 
	\\ 
	\emph{co-hohc} & $D \ ::= A_r \mid G \supset D \mid  D\land D  \ \mid \forall Var \ D$ & $G \ ::= \top \mid A \mid G \land G \mid  G\lor G  \mid \exists Var \ G$ 
	\\ 
	\emph{co-fohh} & $D \ ::= A^1 \mid G \supset D \mid  D\land D  \ \mid \forall Var \ D$ & 
	$G \ ::= \top \mid A^1 \mid G \land G \mid  G\lor G  \mid \exists Var \ G \mid D \supset G\mid \forall Var\ G$ 
	\\ 
	\emph{co-hohh} & $D \ ::= A_r \mid G \supset D \mid  D\land D  \ \mid \forall Var \ D$ & $G \ ::= \top \mid A \mid G \land G \mid  G\lor G  \mid \exists Var \ G \mid D \supset G\mid \forall Var\ G$ 
	\\ \hline
\end{tabular}}
}
\caption{\footnotesize{\textbf{D- and G-formulae}. For each language, we define by mutual induction \emph{goal formulae} (denoted by $G$) and \emph{program clauses} (denoted by $D$). 
}}
\label{tab:up}
\end{table*}
  
Given a signature $\Sigma$, a \emph{logic program} $P$ is a finite set of $D$-formulae over $\Sigma$.
	In the rest of this paper, we will assume that the signature is determined by a given logic program $P$ and is \emph{always first-order}. 
	We can still build higher-order terms (including guarded fixed point terms) over a first-order signature as we place no restriction   on types of variables. 
A \emph{sequent} is an expression of the form $\Sigma;P\longrightarrow G$, encoding the proposition that the goal formula $G$ is 
provable from the logic program $P$ based on the signature $\Sigma$. We recall the \emph{uniform proof} rules 
in Figure~\ref{fig: inductive uniform proof}. 

\begin{figure*}[t]
\centering
\footnotesize{
	$ \infer[\textsc{$\supset\!\! R$}]{\Sigma;P \longrightarrow D \supset G}{\Sigma;P, D \longrightarrow G} $ \hspace{1cm}
	$ \infer[\textsc{$\forall R$}]{\Sigma; P \longrightarrow \forall_{\!\tau} x\ G}{c:\tau,\Sigma; P \longrightarrow G\left[x:=c\right]}$ \hspace{1cm}
	$\infer[\textsc{$\exists R$}]{\Sigma;P \longrightarrow \exists_\tau x\ G}{\Sigma;P \longrightarrow G\left[x:=N\right] }$ 
		
	\vspace{7pt}	
	
	$ \infer[\textsc{$\lor R$}]{\Sigma;P \longrightarrow G_1 \lor G_2}{\Sigma;P \longrightarrow G_1 }\quad
	\infer[\textsc{$\lor R$}]{\Sigma;P \longrightarrow G_1 \lor G_2}{\Sigma;P \longrightarrow G_2 }$ \hspace{1cm} 
	$\infer[\textsc{$\land R$}]{\Sigma;P \longrightarrow G_1 \land G_2}{\Sigma;P \longrightarrow G_1 & \Sigma;P \longrightarrow G_2} $ 
	
	\vspace{5pt}
	 
	$\infer[\textsc{$\supset\! L$}]{\Sigma;P \stackrel{G \supset D}{\longrightarrow} A}{\Sigma;P \stackrel{D}{\longrightarrow} A & \Sigma;P \longrightarrow G }$ \hspace{1cm}
	$\infer[\textsc{$\land L$}]{\Sigma;P \stackrel{D_1 \land D_2}{\longrightarrow} A}{\Sigma;P \stackrel{D_1}{\longrightarrow} A}\quad
	\infer[\textsc{$\land L$}]{\Sigma;P \stackrel{D_1 \land D_2}{\longrightarrow} A}{\Sigma;P \stackrel{D_2}{\longrightarrow} A}$ \hspace{1cm}
	$\infer[\textsc{$\forall L$}]{\Sigma;P \stackrel{\forall_{\!\tau} x\ D}{\longrightarrow} A}{\Sigma;P \stackrel{D\left[x:=N\right]}{\longrightarrow} A }$ 
	
	\vspace{5pt}
	
		$ \infer[\textsc{decide}]{\Sigma; P \longrightarrow A}{\Sigma;P \stackrel{D}{\longrightarrow} A}$\hspace{1cm} 
		$ \infer[\textsc{initial}]{\Sigma; P \stackrel{A'}{\longrightarrow} A}{}$ \hspace{1cm}
		$\infer[\textsc{$\top R$}]{\Sigma;P \longrightarrow \top}{}\quad$ 
		}
	\caption[Uniform Proof]{\footnotesize{\textbf{Uniform proof rules~\cite{MN12}}. A sub-expression $P,D$ on the left hand side of a sequent arrow stands for $P\cup\{D\}$. 
		\emph{Rule restrictions:} in $\exists R$ and $\forall L$, $\Sigma;\emptyset\vdash N:\tau$. Moreover, 
			if used in  \emph{co-fohc} or \emph{co-fohh}, then $N$ is first order;
			if used in  \emph{co-hohc}, then $N\in {\mathcal{U}_1^\Sigma}$; 
			if used in  \emph{co-hohh}, then  $N\in\mathcal{U}_2^\Sigma$.
		In $\forall R$,  $c:\tau\notin\Sigma$ and the expression $c:\tau,\Sigma$ denotes $\Sigma \cup \{c:\tau\}$ ($c$ is also known as an \emph{eigenvariable}).
		In $\textsc{decide}$,  $D\in P$. 
		If the rule  \textsc{initial} is used in \emph{co-fohc} or \emph{co-fohh}, then $A\equiv A'$; if it is used in \emph{co-hohc} or \emph{co-hohh}, then $A=_{fix\beta} A'$.  
	}}\label{fig: inductive uniform proof}
\end{figure*}

We can now proceed to introduce the main rule of this paper, i.e. the coinductive proof rule. We take inspiration from~\cite{Coq94,Gimenez98}, but make a few adaptations to the uniform proof syntax. 
We distinguish  \emph{coinductive} entailment from entailment by the original uniform rules. An expression of the form $\Sigma;P\looparrowright  G$, also called a \emph{sequent}, means: the goal formula $G$ is \emph{coinductively} provable from the logic program $P$ on the signature $\Sigma$. The \textsc{co-fix} rule for $\looparrowright$  is given in  Figure~\ref{fig: general co-fix}.

\begin{figure}[H] 
\footnotesize{
\begin{center}
	$\infer[\textsc{co-fix}]{\Sigma;P \looparrowright M}{\Sigma;P,M \longrightarrow \langle M\rangle}$
\end{center}
}
	\caption[Co-fix rule]{\footnotesize{\textbf{Coinductive fixed point rule.} 
		In the upper sequent of \textsc{co-fix} rule, the left occurrence of $M$ is called a \emph{coinductive hypothesis}, and the right occurrence of $M$, marked by $\langle\rangle$, is called a \emph{coinductive goal}.}}\label{fig: general co-fix}
\end{figure}


This simple step already gives a useful insight into suitable coinductive proof principles for the four logics we introduced. Namely, because the same formula $M$ must occur both as a \emph{coinductive hypothesis} on the left of $\longrightarrow$ and the \emph{coinductive goal} on the right of $\longrightarrow$ in the rule \textsc{co-fix}, this restricts the syntax of the coinductive hypothesis to \emph{core formulae}, i.e. formulae satisfying the definition of both program clauses and goals:\\
\begin{center}
{\footnotesize{\renewcommand{\arraystretch}{1.5}
 \begin{tabular}{p{1.5cm}  p{4cm}  p{1.5cm}  p{5.5cm} } 
	 \emph{co-fohc:} & $M := A^1 \mid M\land M$ & \emph{co-fohh:}  & $ M := A^1 \mid M \land M \mid M \supset M \mid \forall Var\ M$ \\
	\emph{co-hohc:} & $M := A_r \mid M \land M$  & \emph{co-hohh:} &  $M :=  A_r \mid M \land M \mid M \supset M \mid \forall Var\ M$
\end{tabular} }}
\end{center}

To the best of our knowledge such precise relationship between the power of the calculus in question and the expressivity of the coinductive hypotheses it admits has not been described in the literature before. Coq, for example, imposes some restrictions on the shape of coinductive hypotheses (e.g. existential formulae cannot be taken as coinductive hypotheses in Coq), but those restrictions are motivated differently.

The next problem is how to guard the coinductive rule from unsound applications.  To this aim, we introduce the guarding notation $\langle M\rangle$ within the \textsc{co-fix} rule, signifying the guarded status of this goal.
In order to be able to construct proofs with guarded formulae, we introduce a set of rules in Figure~\ref{fig: rules with angles}.
A similar approach to guarding coinductive proofs has recently been suggested independently (and in a different calculus) by Basold~\cite{B18}. 

The \emph{coinductive uniform proof rules} 
are given in 
Figures~\ref{fig: inductive uniform proof}, \ref{fig: general co-fix}, \ref{fig: rules with angles}.
A  \emph{proof} for a sequent 
is a finite tree constructed using coinductive uniform proof rules and such that the root is labelled with $\Sigma;P\looparrowright M$ 
or $\Sigma;P\longrightarrow G$, and leaves are labelled with \emph{initial sequents} (i.e. sequents that can occur as a lower sequent in the rules \textsc{initial} or $\textsc{initial}\langle\rangle$). We say that a proof is constructed in \emph{co-fohc}, \emph{co-fohh}, \emph{co-hohc} or \emph{co-hohh} depending on whether the set of all formulae occuring in its tree satisfy the syntax of \emph{co-fohc}, \emph{co-fohh}, \emph{co-hohc} or \emph{co-hohh}.

 \begin{figure*}[t]
 	\footnotesize{
 	$\infer[\textsc{$\supset\!\! R$}\langle\rangle]{\Sigma;P \longrightarrow \langle M_1 \supset M_2\rangle}{\Sigma;P, M_1 \longrightarrow \langle M_2\rangle} $
 	\ \ \
 	$\infer[\textsc{$\forall R$}\langle\rangle]{\Sigma; P \longrightarrow \langle \forall_{\!\tau} x\ M\rangle}{c:\tau,\Sigma; P \longrightarrow \langle M\left[x:=c\right]\rangle}$
   	 	\ \ \
		 	$\infer[\textsc{$\land R$}\langle\rangle]{\Sigma;P \longrightarrow \langle M_1 \land M_2\rangle}{\Sigma;P \longrightarrow \langle M_1\rangle & \Sigma;P \longrightarrow \langle M_2\rangle} $
 		
 		\vspace{5pt}
 			
 	 	$\infer[\textsc{$\supset\! L$}\langle\rangle]{\Sigma;P \stackrel{G \supset D}{\longrightarrow} \langle A\rangle }{\Sigma;P \stackrel{D}{\longrightarrow} A & \Sigma;P \longrightarrow  G }$
	\ \ \
 	 	$\infer[\textsc{$\land L$}\langle\rangle]{\Sigma;P \stackrel{D_1 \land D_2}{\longrightarrow} \langle A\rangle }{\Sigma;P \stackrel{D_1}{\longrightarrow} \langle A\rangle }$
 	 	\ \ \
	$\infer[\textsc{$\land L$}\langle\rangle]{\Sigma;P \stackrel{D_1 \land D_2}{\longrightarrow} \langle A\rangle }{\Sigma;P \stackrel{D_2}{\longrightarrow} \langle A\rangle }$
 	 	\ \ \
 	 	$\infer[\textsc{$\forall L$}\langle\rangle]{\Sigma;P \stackrel{\forall x\ D}{\longrightarrow} \langle A\rangle }{\Sigma;P \stackrel{D\left[x:= N\right]}{\longrightarrow} \langle A\rangle }$}
		 
		 \vspace{5pt}
		 	
		$\infer[\textsc{decide}\langle\rangle]{\Sigma; P \longrightarrow \langle A\rangle}{\Sigma;P \stackrel{D^*}{\longrightarrow} \langle A\rangle }$ 
		\ \ \
		$\infer[\textsc{initial}\langle\rangle]{\Sigma; P \stackrel{A'}{\longrightarrow} \langle A\rangle}{}$
 	\caption{\footnotesize{\textbf{Rules for guarded coinductive goals.} Rule restrictions: $D^*$ in $\textsc{decide}\langle\rangle$ must be a clause from the \emph{original} program, i.e. not a coinductive hypothesis or a formula added by $\supset\!\! R\langle\rangle$.} }\label{fig: rules with angles}
 \end{figure*}

We now demonstrate the coinductive proofs in   these  four calculi  
using the running examples from Introduction. In all examples, we use symbol $P$ to denote the logic  program that conists of  clauses (1) -- (9) given in Introduction. 
When we want to refer to a particular clause $m$ in $P$, we will use notation $P(m)$.

\begin{example}[Coinductive rules in \emph{co-fohc}]\label{ex:2}
Figure~\ref{fig:ex2} shows a proof involving clauses (6) and (7) from Introduction, and not involving infinite stream construction.  The coinductive hypothesis and the goal $member \ 0\ [0\ | \ nil]$ is expressed in \emph{co-fohc}.
Notice how the coinductive goal remains guarded until the application of the rule $\supset\! L \langle \rangle$, after which the coinductive hypothesis can be safely applied.
As usual for coinductive proofs in \emph{co-fohc}, the proof is regular, in the sense that the coinductive hypothesis 
 taken at the start of the proof applies verbatim later in the proof.  
\end{example}

\begin{figure*}[t]
\begin{displaymath}
\resizebox{\textwidth}{!}{
		\infer[\textsc{co-fix}]		
		{
			\Sigma;P \looparrowright \mathit{ member \ 0\ [0\ | \ nil]}
		}
		{	
		\infer[\textsc{decide}\langle \rangle]	
			{
				\Sigma;P,CH\longrightarrow\mathit{\langle   member \ 0\ [0\ | \ nil] \rangle}
			}
			{	
				\infer[\forall L\left \langle \rangle(\text{3 times}\right)]
				{		
					\Sigma;P,CH\stackrel{P\left(6\right)}{\longrightarrow}\mathit{\langle member \ 0\ [0\ | \ nil] \rangle}  					
				}				
				{	
					\infer[\supset L \langle \rangle]	
					{
						\Sigma;P,CH\stackrel{\mathit{(member\ 0\ [0 | nil] \land eq \ 0 \ 0) \;\supset\; member \ 0\ [0 | nil] }}{\longrightarrow}\mathit{\langle member \ 0\ [0\ | \ nil] \rangle}	
					}
					{	
						\infer[\textsc{initial}]	
						{
									\Sigma;P,CH\stackrel{\mathit{member \ 0\ [0 | nil] }}{\longrightarrow}\mathit{ member \ 0\ [0\ | \ nil] }	
						}
						{	}           
						&
					\infer[\land R]
	{
		\Sigma;P,CH\longrightarrow\mathit{ member\ 0\ [0 | nil] \land eq \ 0 \ 0 }}
						{ 
						\infer[\textsc{decide}]{\Sigma;P,CH\longrightarrow\mathit{ member\ 0\ [0 | nil] }}
						{
									\infer[\textsc{initial}]{\Sigma;P,CH\stackrel{CH}{\longrightarrow}\mathit{ member\ 0\ [0 | nil] }}
									{}
						}
						&
						\infer[\textsc{decide}]{\Sigma;P,CH\longrightarrow\mathit{ eq \ 0 \ 0 }}
						{
						\infer[\forall L]{\Sigma;P,CH\stackrel{P(7)}{\longrightarrow}\mathit{ eq \ 0 \ 0 }}
						{
						\infer[\textsc{initial}]{\Sigma;P,CH\stackrel{eq\ 0 \ 0}{\longrightarrow}\mathit{ eq \ 0 \ 0 }}
						{}
						}
						}
						}
}
																			}
						}	}	
			}
\end{displaymath}
	
\caption{ 	\footnotesize{\textbf{Running example: a coinductive proof for the sequent $\Sigma;P \looparrowright \mathit{member\ 0\ [0 \ | \ nil]}$ in \emph{co-fohc}.} We use $CH$ to abbreviate the coinductive hypothesis $\mathit{member\ 0\ [0 \ | \ nil]}$. }}\label{fig:ex2}
	
\end{figure*}

\begin{example}[Coinductive rules in \emph{co-hohc}]\label{ex:1}
The next proof (for clauses (3) and (4)) is given in Figure~\ref{fig:ex1}. It is similar to the proof of Figure~\ref{fig:ex2}, in that it is also regular, and requires a coinductive hypothesis satisfying the syntax of Horn clause logic.
However, this time the coinductive hypothesis  $bitstream\ [0 \ | \ n\_str \ 0 ]$ involves construction of a stream of zeros, given as $n\_str \ 0$, where the fixed point term $n\_str$ is defined in Example~\ref{eg: number stream as fix point term}. In Introduction, we actually wanted to prove
$\exists \ y (bistream\ [0| y])$, but existential formulae are not allowed to be coinductive hypotheses, by the previous discussion. 
We can only prove $\exists \ y (bistream\ [0| y])$ if we add the coinductively proven lemma   $bitstream\ [0 \ | \ n\_str \ 0 ]$ to  $P$. Thus, we can infer $\Sigma; P, bitstream\ [0 \ | \ n\_str \ 0 ] \longrightarrow \exists \ y (bistream\ [0| y])$ by the standard rules
of \emph{hohc} given in Figure~\ref{fig: inductive uniform proof}.
\end{example}
In the next section, we will prove that such manipulation with proven lemmas is sound. 
Our proof can be seen as a semantic version of cut admissibility. However, we follow the uniform proof tradition and do not introduce a cut rule directly into the calculus.

\begin{figure*}[t]
\begin{displaymath}
\resizebox{\textwidth}{!}{
		\infer[\textsc{co-fix}]		
		{
			\Sigma;P \looparrowright \mathit{bitstream\ [0 \ | \ n\_str \ 0 ]}
		}
		{	
		\infer[\textsc{decide}\langle \rangle]	
			{
				\Sigma;P,CH\longrightarrow\mathit{\langle   bitstream\ [0 \ | \ n\_str \ 0 ] \rangle}
			}
			{	
				\infer[\forall L\left \langle \rangle(\text{2 times}\right)]
				{		
					\Sigma;P,CH\stackrel{P\left(3\right)}{\longrightarrow}\mathit{\langle bitstream\ [0 \ | \ n\_str \ 0 ] \rangle}  					
				}				
				{	
					\infer[\supset L \langle \rangle]	
					{
						\Sigma;P,CH\stackrel{\mathit{(bitstream\ (n\_str \ 0) \;\land\; bit\ 0  ) \;\supset\; bitstream [ 0\ |\ n\_str \ 0] }}{\longrightarrow}\mathit{\langle bitstream\ [0 \ | \ n\_str \ 0 ] \rangle}	
					}
					{	
						\infer[\textsc{initial}]	
						{
									\Sigma;P,CH\stackrel{\mathit{bitstream\ [0 \ | \ n\_str \ 0 ] }}{\longrightarrow}\mathit{ bitstream\ [0 \ | \ n\_str \ 0 ] }	
						}
						{	}                
						&
	\infer[\land R]
	{
		\Sigma;P,CH\longrightarrow\mathit{ bitstream\ (n\_str \ 0) \;\land\; bit\ 0  }}
	{ 
		\infer[\textsc{decide}]{\Sigma;P,CH\longrightarrow\mathit{ bitstream\ (n\_str \ 0) }}
		{
		\infer[\textsc{initial}\checkmark]{\Sigma;P,CH\stackrel{CH}{\longrightarrow}\mathit{ bitstream\ (n\_str \ 0)}}
		{}
		}
						&
		\infer[\textsc{decide}]{\Sigma;P,CH\longrightarrow\mathit{ bit\ 0  }}
		{
		  \infer[\textsc{initial}]{\Sigma;P,CH\stackrel{P(4)}{\longrightarrow}\mathit{ bit\ 0  }}
		  {
		   }
		 }
	 }
    }
													}
	}			}
				}
\end{displaymath}

\caption{	\footnotesize{\textbf{Running example: a coinductive proof for the sequent $\Sigma;P \looparrowright \mathit{bitstream\ [0 \ | \ n\_str \ 0 ]}$ in \emph{co-hohc}.} We use $CH$ to abbreviate the coinductive hypothesis $\mathit{bitstream\ [0 \ | \ n\_str \ 0]}$. Note that, at the step marked by $\checkmark$,  $CH=_{fix\beta}\mathit{ bitstream\ (n\_str \ 0)}$ }}\label{fig:ex1}
	
\end{figure*}


\begin{example}[Coinductive rules in \emph{co-hohh}]\label{ex:3}
The next example of a coinductive uniform proof (for clause (8)) is given in Figure~\ref{fig:ex3}.  It is again more complicated than the previous examples, in that it requires not just fixed point terms, but also the syntax of hereditary Harrop logic for its coinductive hypothesis and goal (given by $\forall x (from\ x\ (fr\_str\ x))$). To see this, suppose we tried to prove $from\ 0\ (fr\_str\ 0)$ directly, following Example~\ref{ex:1}, and suppose we took the coinductive hypothesis $from\ 0\ (fr\_str\ 0)$. After resolving with the clause (8) via the rule $\supset\! L \langle \rangle$, we would reach a goal $from\ (s\ 0\ (fr\_str\ (s \ 0))$ to which the given coinductive hypothesis would not apply.
The reason is irregularity  of the stream in question, that requires us to prove a more general coinductive goal.
As in the previous example, if we wanted to obtain a proof for $from\ 0\ (fr\_str\ 0)$, we can derive it by the rules of Figure~\ref{fig: inductive uniform proof} from $P \cup \forall x (from\ x\ (fr\_str\ x))$. 
\end{example}

\begin{figure*}[t]	

\begin{displaymath}
	\resizebox{\textwidth}{!}{
	\infer[\textsc{co-fix}]		
	{
		\Sigma;P \looparrowright \mathit{\forall x (from\ x\ (fr\_str\ x))}
	}
	{	
		\infer[\forall R \langle \rangle]
		{\Sigma;P,CH\longrightarrow\mathit{\langle \forall x (from\ x\ (fr\_str\ x)) \rangle }}
		{
			\infer[\textsc{decide}\langle \rangle]	
			{
				Z,\Sigma;P,CH\longrightarrow\mathit{\langle  from\ Z\ (fr\_str\ Z) \rangle}
			}
			{	
				\infer[\forall L\left \langle \rangle(\text{2 times}\right)]
				{		
					Z,\Sigma;P,CH\stackrel{P\left(8\right)}{\longrightarrow}\mathit{\langle from\ Z\ (fr\_str\ Z) \rangle}  					
				}				
				{	
					\infer[\supset L \langle \rangle]	
					{
						Z,\Sigma;P,CH\stackrel{\mathit{from\ (s\ Z) (fr\_str\ s\ Z)\  \;\supset\; from\ Z\ [ Z\ |\ (fr\_str\ s \ Z) ] }}{\longrightarrow}\mathit{\langle from\ Z\ (fr\_str\ Z) \rangle}	
					}
					{	
						\infer[\textsc{initial}\checkmark]	
						{
							Z,\Sigma;P,CH\stackrel{\mathit{from\ Z\ [ Z\ |\ (fr\_str\ (s \ Z)) ] }}{\longrightarrow}\mathit{from\ Z\ (fr\_str\ Z)}
						}
						{}                 
						&
						\infer[\textsc{decide}]
						{
							Z,\Sigma;P,CH\longrightarrow \mathit{from\  (s\ Z)\  (fr\_str\ (s \ Z))}	}
						{ 
							\infer[\forall L]
							{
							Z,\Sigma;P,CH \stackrel{CH}{\longrightarrow} \mathit{from\  (s\ Z)\  (fr\_str\ (s \ Z))}
							}
							{
								\infer[\textsc{initial}]
								{Z,\Sigma;P,CH \stackrel{\mathit{from\ (s \ Z)\ (fr\_str\ (s \ Z))}}{\longrightarrow} \mathit{from\  (s\ Z)\  (fr\_str\ (s \ Z))}}
								{}
							}	
						}
					}
				}
			}
		}			
	}
} 
	\end{displaymath}			
\caption{	\footnotesize{\textbf{Running example: a coinductive proof for the sequent $\Sigma;P \looparrowright \forall x (\mathit{from\ x\ (fr\_str\ x))}$ in \emph{co-hohh}}, where $fr\_str$ is defined in Example~\ref{eg: from as fix point}, and $Z$ is an arbitrary eigenvariable. $CH$ abbreviates the coinductive hypothesis $\mathit{\forall x (from\ x\ (fr\_str\ x))}$. The step marked by $\checkmark$ indicates involvement of the relation $\mathit{from\ Z\ [ Z\ |\ (fr\_str\ (s \ Z)) ]}	=_{\textit{fix}\beta} \mathit{from\ Z\ (fr\_str\ Z)}$. }}\label{fig:ex3}
	
\end{figure*}

\begin{example}[Coinductive rules in \emph{co-fohh}]\label{ex:4}
Finally, to complete the picture, we give an example of a proof in \emph{co-fohh}  in Figure~\ref{fig:ex4}.
This example uses clause (9) and does not require infinite data construction via fixed-point terms, but unlike all other examples, it shows
that implicative coinductive hypotheses 
may play an important role.
If we wanted to coinductively prove that a bit $0$  and some given stream, say $n\_str\ 0$, satisfy the relation  $comemember_{bit}$,
we would not be able to prove it directly: similarly to the case of Example~\ref{ex:3}, the proof is irregular. I.e., taking the coinductive hypothesis and  goal $comemember_{bit}\ 0 \ (n\_str\ 0)$, we would resolve it with the clause (9), only to get a subgoal 
$comemember_{bit}\ 0 \ f(n\_str\ 0)$, to which the given coinductive hypothesis would not apply. 
But we would be able to prove $comemember_{bit}\ 0 \ (n\_str\ 0)$, if we coinductively prove $\forall \ y \ s \ (bit\ y \ \supset comemember_{bit}\ y \ s)$ first, as Figure~\ref{fig:ex4} shows,  and then use the extended logic program 
$P \cup \forall \ y \ s \ (bit\ y \ \supset comemember_{bit}\ y \ s)$ and the rules of Figure~\ref{fig: inductive uniform proof}.
\end{example}

\begin{figure*}[t]
\begin{displaymath}
\resizebox{\textwidth}{!}{
\infer[\textsc{co-fix}]		
{
	\Sigma;P \looparrowright \mathit{\forall \ y \ s \ (bit\ y \ \supset com_{bit}\ y \ s)}
}
{	
	\infer[\forall R \langle \rangle\ (\text{twice})]
	{
		\Sigma;P,CH\longrightarrow\mathit{\langle   \forall \ y \ s \ (bit\ y \ \supset com_{bit}\ y \ s) \rangle}
	}
	{
		\infer[\supset R \langle \rangle]
		{
			B,S,\Sigma;P,CH\longrightarrow\mathit{\langle   bit\ B \ \supset com_{bit}\ B \ S \rangle}
		}
		{
		\infer[\textsc{dec}\langle \rangle]
		{
			B,S,\Sigma;P,CH, bit\ B \longrightarrow\mathit{\langle   com_{bit}\ B \ S \rangle}
		}
		{	
			\infer[\forall L\left \langle \rangle\ (\text{twice}\right)]
			{		
			B,S,\Sigma;P,CH, bit\ B \stackrel{P\left(9\right)}{\longrightarrow}\mathit{\langle com_{bit}\ B \ S \rangle}
		    }				
			{	
				\infer[\supset L \langle \rangle]
				{
				B,S,\Sigma;P,CH, bit\ B \stackrel{\mathit{(com_{bit} \ B\ (f \ S )\land bit \ B ) \;\supset\; com_{bit} \ B \ S }}{\longrightarrow}\mathit{\langle  com_{bit}\ B \ S \rangle}
			    }
				{	
				  \infer[\textsc{init}]
				  {
				  	B,S,\Sigma;P,CH, bit\ B \stackrel{\mathit{com_{bit}\ B \ S}}{\longrightarrow}\mathit{ com_{bit}\ B \ S }
			  	  }
						{	}           
						&
	\infer[\land R]
	{
	B,S,\Sigma;P,CH, bit\ B \longrightarrow\mathit{ com_{bit} \ B\ (f \ S ) \land bit \ B  }
	}
	{ 
		\infer[\textsc{dec}]
		{
		B,S,\Sigma;P,CH, bit\ B \longrightarrow\mathit{ com_{bit} \ B\ (f \ S ) }
	    }
		{
		\infer[\forall L \ (\text{twice})]
		{
		B,S,\Sigma;P,CH, bit\ B \stackrel{CH}{\longrightarrow}\mathit{ com_{bit} \ B\ (f \ S )}
		}
		{
		\infer[\supset L]
		{
		B,S,\Sigma;P,CH, bit\ B \stackrel{\mathit{bit\ B \supset com_{bit}\ B \ (f \ S )}}{\longrightarrow}\mathit{ com_{bit} \ B\ (f \ S ) }
	    }
		{
		\infer[\textsc{init}]
		{
		B,S,\Sigma;P,CH, bit\ B \stackrel{\mathit{com_{bit}\ B \ (f \ S )}}{\longrightarrow}\mathit{ com_{bit} \ B\ (f \ S ) }
	    }{} 		
					&
		\infer[\textsc{dec}]
		{B,S,\Sigma;P,CH, bit\ B \longrightarrow \mathit{ bit\ B }}		
		{
		\infer[\textsc{init}]
		{B,S,\Sigma;P,CH, bit\ B \stackrel{\mathit{ bit\ B }}{\longrightarrow} \mathit{ bit\ B }}	
					{}
					}
					}}}
					& {Proof^*}
						}
						}
            }
						}	}		} } }
\end{displaymath}
	
\caption{\footnotesize{\textbf{Running example: a coinductive proof for the sequent $\Sigma;P \looparrowright \mathit{\forall \ y \ s \ (bit\ y  \supset comember_{bit}\ y \ s)}$ in \emph{co-fohh}.}   We abbreviate the coinductive hypothesis $\mathit{\forall \ y \ s \ (bit\ y \supset comember_{bit}\ y \ s)}$ by CH, the rules INITIAL and DECIDE by INIT and DEC, and $comember_{bit}$ --- as $com_{bit}$. $Proof^*$ is a copy of the proof for $B,S,\Sigma;P,CH, bit\ B \longrightarrow \mathit{ bit\ B}$ given also in another branch of the proof.  $B$ and $S$ are arbitrary eigenvariables. }}\label{fig:ex4}
	
\end{figure*}

Although the last example looks somewhat artificial in the context of this paper, examples of similar proofs do arise in automated proving with Horn clauses, and in particular in implementation of Haskell type classes, as~\cite{FKS15,BottuKSOW17} report.
Appendix~\ref{sec:FP} gives an  example of a more complex proof in \emph{co-fohh}, taken directly from~\cite{FKS15}.

\section{Soundness of Coinductive Uniform Proofs for Logic Programs}\label{sec:sound}

We start this section by recalling the standard definitions of coinductive   models for logic programs~\cite{Llo87}.
The first step in definition of such models  is to define infinite
 tree-terms and tree-atoms that inhabit such models. We follow standard definitions in this regard~\cite{Courcelle83}, see Appendix~\ref{sec:trees}.
Informally, a tree-term is defined as a map from a set of lists of non-negative integers into $\Sigma$, with tree branching respecting the term arity. 
If the domain of the map is infinite, the tree-term is infinite.
This definition then extends to tree-atoms in an obvious way.

 Given a first order signature $\Sigma$ and a logic program $P_\Sigma$, the 
\emph{coinductive Herbrand universe} of $P_{\Sigma}$, denoted $\mathcal{H}^\Sigma$, is the set of all finite and infinite closed tree-terms on $\Sigma$.
The \emph{coinductive Herbrand base} of $P_{\Sigma}$, 
denoted $\mathcal{B}^\Sigma$, is the set of all finite and infinite closed tree-atoms on $\Sigma$.

We first establish a connection between guarded atoms on $\Sigma$ and
the coinductive Herbrand base of $P_{\Sigma}$. It is a form of  productivity result for guarded fixed-point terms.

If $A$ is a first order atom, then we will denote the equivalent tree-atom by $A^T$. Let $A$ be a guarded atom on $\Sigma$ and let $\diamond:\iota\notin\Sigma$, then a \emph{snapshot} of $A$, denoted $A\diamond$, is an atom obtained 
by replacing all guarded full terms in $A$ with $\diamond$, and it is understood that a guarded full term $N$ in $A$ is replaced by $\diamond$ only if $N$ is \emph{not} a first order term. Clearly, $A\diamond$ is a first order atom. A guarded atom has a \emph{fair} infinite sequence of $\textit{fix}\beta$-reductions if every $\textit{fix}\beta$-reducible sub-term is reduced within a finite number steps. The next lemma relies on a standard definition of a metric $d$ on tree-terms from~\cite{Llo87}. Given $t,s$ as two tree-terms (tree-atoms), $d(t,s)$ denotes the distance between $t,s$, where $0\leq d(t,s) < 1$ and the smaller $d(t,s)$ is, the more similar $t,s$ are (cf.  Appendix~\ref{sec:trees}).

\begin{lemma}[Productivity lemma]\label{lem:prod}
Let $A$ be a guarded atom on a first-order signature $\Sigma$ and $\diamond:\iota\notin\Sigma$. If $A$ has a fair infinite sequence of  one step $\textit{fix}\beta$-reductions $A \rightarrow_{\textit{fix}\beta}  A'_1,\   A_1\rightarrow_{\textit{fix}\beta}  A'_2,\  A_2\rightarrow_{\textit{fix}\beta} A'_3, \ldots$, where $A_k$ is the $\beta$-normal form of $A'_k$, then there exist an infinite tree-atom $A^T$ such that $d((A_k\diamond)^T, A^T) \to 0$ as $k\to\infty$. 
Moreover, if $A$ is a closed guarded atom, then $A^T\in\mathcal{B}^\Sigma$. 
\end{lemma}
\begin{proof}
Similar to the proof in Komendantskaya and Li~ \cite[Lemma 4.1]{KL17}.
\end{proof}

Lemma~\ref{lem:prod} allows us to extend the notation $A^T$, from requiring $A$ be a first order atom, to allowing $A$ to be any guarded atom, and shows that $A^T\in \mathcal{B}^\Sigma$ if $A$ is closed.

We now proceed to consider coinductive models of logic programs.
If a $D$-formula in \emph{(co)-hohc} is in the 
form \[\forall\!_{\iota}x_1\ldots x_m\quad A_1\land \ldots \land A_n \supset A\quad \left(m,n\geq 0\right) \] then it is also called an \emph{$H$-formula}, which is perhaps a better known presentation of Horn clauses.
It is well known that in both classical and intuitionistic logics, a set of Horn clause $D$-formulae can be transformed into an equivalent set of $H$-formulae, and \emph{vice versa}~\cite[\protect\S 2.6.2]{MN12}. 
Also note that, by the discussion of the previous section, this is the only kind of formulae that we can coinductively prove using the \textsc{co-fix} rule. 

 Given an $H$-formula $K=_{\mathit{def}}\forall\!_\iota x_1\ldots x_m\; A_1\land \ldots \land A_n \supset A$, we define:\\
	%
	%
{\footnotesize{\renewcommand{\arraystretch}{1.5}
 \begin{tabular}{p{6.5cm} || p{6.5cm} } 
\hline
	 \emph{tree-form ground instance} $\lfloor K \rfloor^T$:  & \emph{term-form ground instance} $\lfloor K \rfloor$:  \\
	\hline
	$\left(A^T_1\land \ldots \land A^T_n \supset A^T\right)\left[x_1:=N_1\right]\cdots\left[x_m:=N_m\right]$ & $\left(A_1\land \ldots \land A_n \supset A\right)\left[x_1:=N'_1\right]\cdots\left[x_m:=N'_m\right]$\\
	where each $N_k\in{\mathcal{H}^\Sigma}$ & where each $N'_k\in{\mathcal{U}_1^\Sigma}$\\ \hline
\end{tabular} }}\\
%
If $K'=_{\mathit{def}}{A_1'\land \ldots \land A_n' \supset A'}$ is a ground instance (in either tree-form or term-form), then we denote $A'$ by $\mathit{head}\ K$ and denote the set $\{A_1',\ldots, A_n'\}$ by $\mathit{body}\ K$.

A \emph{Herbrand interpretation}, denoted $I$, is any subset of ${\mathcal{B}^\Sigma}$. We denote a powerset of $\mathcal{B}^\Sigma$ by $\mathit{Pow}\left(\mathcal{B}^\Sigma\right)$. It is well known that 
${\langle\mathit{Pow}\left(\mathcal{B}^\Sigma\right), \subseteq\rangle}$ is a complete lattice. 
Appendix~\ref{sec:lattice} recalls some standard lattice-theoretic definitions and results we use. 
The \emph{immediate consequence operator} with respect to a logic program $P_{\Sigma}$,  
$\mathcal{T}:\mathit{Pow}\left(\mathcal{B}^\Sigma\right)\mapsto\mathit{Pow}\left(\mathcal{B}^\Sigma\right)$, is defined as
\begin{align*}
\mathcal{T}\left(I \right)= & \{B\in\mathcal{B}^\Sigma\mid  F\in P_\Sigma,\ \mathit{head}\ \lfloor F \rfloor^T = B,\ \mathit{body}\ \lfloor F \rfloor^T\subseteq I\ \}
\end{align*}

A Herbrand interpretation $I$ is a model of  $P_\Sigma$ if and only if $I$ is a pre-fixed point of $\mathcal{T}$.

 Using the fact that $\mathcal{T}$ is increasing~\cite{Llo87}, we can rely on the Knaster-Tarski theorem (c.f. Appendix~\ref{sec:lattice}) to assert that its greatest fixed point exists:
\begin{alignat*}{4}
\mathcal{M}  &=  \mathit{gfp}\left(\mathcal{T}\right) & ={\bigcup}\{I\mid I= \mathcal{T}\left(I\right) \} &={\bigcup}\{I\mid  I \subseteq  \mathcal{T}\left(I\right) \} 
\end{alignat*}
We say that  
$\mathcal{M}$ is a \emph{coinductive  model} of $P$. 



\noindent In order to use these models for our main soundness results, we first need to formulate a coinductive principle for the proofs involving these models.
The implication from right to left in Lemma 
\ref{lem: coin prf pcpl infi} is an instance of the \emph{coinductive proof principle}, as formulated e.g. in Sangiorgi~\cite[\S 2.4]{Sangiorgi:2011:IBC:2103603}: 



\begin{lemma}[Coinductive proof principle for coinductive models]\label{lem: coin prf pcpl infi}
	Let $P_{\Sigma}$ be a logic program, with the operator $\mathcal{T}$ and 
	the model $\mathcal{M}$. Given a set $S$,  
	$S\subseteq\mathcal{M}$, if and only if,  there exist a  Herbrand interpretation  $I$ for $P_{\Sigma}$, such that $S\subseteq I$ and $I$ is a post-fixed point of $\mathcal{T}$.  
\end{lemma}
	The proof is given in Appendix~\ref{sec:sproof}.

We are now ready to formulate soundness of coinductive uniform proofs. We prove the result for \emph{co-hohh}, but, because proofs in \emph{co-fohc}, \emph{co-fohh} are \emph{co-hohc} are, by definition, also proofs in \emph{co-hohh}, their soundness
relative to coinductive models follows as a corollary of this theorem.

\begin{theorem}[Soundness of  \protect\textit{co-hohh} proofs]\label{thm: main soundess infinite}
	Let
		 $P_{\Sigma}$ be a logic program with a coinductive model $\mathcal{M}$,
		 ${\forall\!_{\iota}x_1\ldots x_m\ A_1\land \ldots \land A_n \supset A}$ be a $H$-formula, and 
		 $\Sigma;P\looparrowright \forall\!_{\iota}x_1\ldots x_m\ A_1\land \ldots \land A_n \supset A$ have a  proof in \textit{co-hohh} which involves only guarded atoms. 
	Then, for an arbitrary tree-form ground instance ${A'_1\land \ldots \land A'_n \supset A'}$ of the goal, $\{A'_1,\ldots A'_n\}\subseteq\mathcal{M}$ implies $\{A'\}\subseteq\mathcal{M}$.
\end{theorem}

\begin{proof}[Sketch]
We use Lemma~\ref{lem: coin prf pcpl infi} from right to left, which means, to show that $\{A'\}\subseteq \mathcal{M}$, we look for a set  $I$ such that the \emph{requirements}
${\{A'\}\subseteq I}$ and
${I \subseteq \mathcal{T}\left(I\right)}$
 are satisfied. The proof follows an \emph{Analysis--Construction--Verification} structure, where we first study the proof of the \emph{root sequent} $\Sigma;P\looparrowright \forall\!_{\iota}x_1\ldots x_m\ A_1\land \ldots \land A_n \supset A$ which provides information for constructing a candidate set $I$. 
In this construction, we use Lemma~\ref{lem:prod}.
Finally we verify that the set $I$ so constructed satisfies the \emph{requirements}.
The exact details of these three steps are given in Appendix~\ref{sec:sproof}.
\end{proof}

Finally, we show that extending logic programs with coinductively proven lemmas is sound.

\begin{theorem}[Conservative model extension]
	\label{them: cut theorem infi model}
	Let a
		logic program $P_\Sigma$ have coinductive model $\mathcal{M}$, and a
		sequent $\Sigma;P\looparrowright H$ have a \textit{co-hohh} proof that only involves guarded atoms.
		Let $H_1,\ldots, H_n$ be distinct term-form ground instances of $H$ involving only guarded atoms, and such that for each $H_k$ ($1\leq k\leq n$ ),   if  
		 $A\in\mathit{body}\, H_k$, then  $A^T\in\mathcal{M}$. Let $P\cup\{H_1,\ldots, H_n\}$ have a coinductive model $\mathcal{M}'$.
	Then, $\mathcal{M}=\mathcal{M}'$.
\end{theorem}
\begin{proof}[Sketch]
	The proof shows equality of the two sets by proving that they are each other's subset. Theorem~\ref{thm: main soundess infinite} and Lemma~\ref{lem: coin prf pcpl infi} are used. The exact details are given in Appendix~\ref{sec:sproof}.
	\end{proof}
By Theorem~\ref{thm: main soundess infinite}, proofs for logic programs $P$ and $P\cup H$ are sound relative to their coinductive models. But, subject to conditions of Theorem~\ref{them: cut theorem infi model}, the models of these two programs are equivalent. Hence, proofs for $P \cup H$ are coinductively sound relative to the coinductive model of $P$. 

\section{Conclusions, Discussion, Future and Related Work}\label{sec:rw}

We have presented a sound method for proving coinductive properties of theories expressed in 
first-order Horn clause logic. We used the four calculi of the ``uniform proof diamond" 
to distinguish important classes of these coinductive properties. The major division is between classes of properties involving (or not involving)
infinite data structures, 
accommodated by higher-order/first-order division of the diamond. 
Then classification further splits into regular and irregular cases of coinductive reasoning, 
accommodated by Horn clause logic/hereditary Harrop clause logic parts of the diamond. 

The most intriguing future direction is to explore how these results may be applied in automated 
reasoning systems based on first-order logic~\cite{BjornerGMR15}.
 The methods presented here are 
different from most popular coinductive methods, such as~\cite{Coq94,Gimenez98,BlanchetteM0T17}, which are based on higher-order logic and/or dependent type theory with (co)inductive data types. 
%
%
For example, in Coq we would need to implement clauses (3), (6), (8), (9) as coinductive types, each clause would be inhabited by a constant proof term, seen by Coq as a coinductive type constructor (the code is given in Appendix~\ref{sec:Coq}).   Coq's way of ensuring soundness of these proofs is type-checking backed by the guardedness checks at the proof term level. Proof terms that inhabit proven propositions must be guarded by constructors of the coinductive type in question.  In comparison, we  formulated the coinductive calculi  without any notion of (co)inductive data types, or any generation of proof terms. Instead of guarding proof terms, we use rules of Figure~\ref{fig: rules with angles} to guard applications of the coinductive hypotheses.
When working with definitions of infinite streams, we introduce only minimal guardedness checks on fixed point terms. 

In this respect, of particular interest is the clean separation  between coinductive proofs
involving (or not involving)  fixed point terms made by the coinductive diamond. Thus, depending on first-order theory in question, we can now 
choose to work with one or the other kind of coinduction:
if a certain prover's syntax does not admit fixed point terms, there are still \emph{co-fohc} and \emph{co-hohh} available for it.  
This differs from majority of automated coinductive methods~\cite{BlanchetteM0T17,LeinoM14}, that were proposed with (co)inductive data structures in mind.
This paper shows that having coinductive data structures is not at all a pre-requisite for having sound coinductive proofs.

An interesting direction for future work is to generalise our definition of guarded fixed point terms, e.g. benefiting from 
recent new methods by Blanchette et al.~\cite{BlanchetteBL0T17}. Adding data structures explicitly to the syntax of unform proofs is a related task,
and deserves attention. In a  sense, a similar work has already been done in a richer system -- Abella~\cite{BaeldeCGMNTW14}. 
Compared to these richer languages, our current paper helped us to explain and systematise 
the kinds of coinductive reasoning 
available for first-order Horn clauses. 

The coinductive uniform proofs  cannot capture all coinductive properties arising in first-order Horn clause logic. Consider the following logic program defining a Fibonacci stream:

\vspace*{0.05in}
\noindent
(10) $\forall x\ (add\ 0\ x\ x) $\\
(11) $\forall x\ y\ z\ (add\ x\ y\ z  \;\supset\; add \ (s\ x)\ y\ (s\ z) )$\\
(12) $\forall x\ y\ z\ ((add\ x\ y\ z \land fibs \ y\ z\ s) \;\supset\; fibs\ x\ y\ [ x\ | s ] )$
\vspace*{0.05in} 

\noindent The syntax of \emph{co-hohh} would allow us to express the coinductive property $\forall x\ y\ z\ (add\ x\ y\ z\ \supset \ fibs \ x\ y\ [x | fibs\ y\ z])$, where $fibs$ is a guarded fixed-point term (relying on a function $+$). We would even be able to  safely apply the coinductive hypothesis in the proof. However, the problem arises with inductive parts of this proof. It inevitably stumbles upon having to prove and use a relatively trivial property 
$\forall x \ y\ add\ x\ y\ (x+y) $. However, it can only be proven by induction, and the original formulation of uniform proofs does not admit the inductive proof principle~\cite{MN12}. Stating this property as part of the logic program would violate our assumption that the given logic program is first-order ($+$ is defined by $\lambda$-abstraction).
This suggests two future directions.

Firstly, it would be useful to complement our coinductive proof principle with an inductive one. Inductive proof principle has been
introduced in a similar framework in~\cite{McDowellM00,BaeldeCGMNTW14}. With this additional tool at hand, the question of mixing induction and coinduction will arise, and several existing methods, such as those given in~\cite{B18}, may prove useful. 

Alternatively, we could lift our current restrictions and allow the syntax of \emph{co-hohh}  in  logic programs.
Section~\ref{sec:cup}  formulates the four coinductive calculi with this generalisation in mind.
However, to prove soundness of these logics, 
we would need to use a much more sophisticated notion of a coinductive model, which was beyond the scope of this paper. Categorical models for mixed induction and coinduction by Basold~\cite{B18,BasoldG16} seem promising in this respect.

In this paper, we only worked with one definition of fixed-point operator at the term level. A solid body of literature exists on introducing both fixed point and co-fixed point operators (also sometimes known as $\mu$ and $\nu$ operators) to calculi similar to ours:~\cite{BaeldeN12,ba08,BrotherstonS11,B18}. 
We plan to look into these methods when extending the notion of a logic program to \emph{co-hohh}, and/or extending our four calculi with 
the inductive proof principle.
 

An orthogonal, but promising direction would be to give a Curry-Howard interpretation to our calculi, as was done already for \emph{co-fohc} and \emph{co-fohh} in~\cite{FKS15}. There, Horn clauses are seen as types, and terms inhabiting these types are constructed alongside the proof rule applications. 
Note that all four calculi we have introduced here are intuitionistic and hence in principle should allow constructive interpretation.
Extension of \cite{FKS15} to \emph{co-hohc} and \emph{co-hohh} would require dependent types, similar to~\cite{B18,BasoldG16}. This work could help to automate
coinductive proofs in interactive theorem provers (such as Coq or Agda), which are based on constructive type theory. 

Finally, we plan  a more systematic study of syntactic and semantics approaches to the cut rule in presence of (mixed) inductive and coinductive proofs,
in order to better relate results of Section~\ref{sec:sound} to the existing literature on the subject~\cite{McDowellM00}.   



{\footnotesize{
\bibliographystyle{plain}
\bibliography{katya2}}}

\pagebreak

\appendix
\section{Omitted Standard Syntactic Definitions}\label{sec:syntax}

In this section, we complete some of the omitted standard syntactic definitions omitted in Section~\ref{sec:backgr}.
The textbook~\cite{MN12} gives the full details on \emph{uniform proofs}.

\paragraph{Free Variables}
 \begin{enumerate}
 \item $FV\left(x\right)= \{x\}$ for all $x\in \mathit{Var}$
 \item $FV\left(c\right)= \emptyset$ for all $c\in \mathit{Cst}$
 \item $FV\left(M\ N\right)= FV\left(M\right)\cup FV\left(M\right)$ for all $M,N\in \Lambda$
 \item $FV\left(\lambda x\, .\,  M\right)= FV\left(M\right)\backslash \{x\}$ for all $M\in \Lambda$ and $x\in\mathit{Var}$
 \item $FV\!\left(\textit{fix}\ M\right)=FV\!\left(M\right)$
 \end{enumerate}

\paragraph{Typing Rules}
 If $\Sigma$ maps constant $c$ to type  $\tau$, we write $c:\tau\in \Sigma$. Moreover we write $\Sigma,c:\tau$ (or $c:\tau,\Sigma$), to denote $\Sigma \cup \{c:\tau\}$. Similar notations apply to $\Gamma$.
\begin{center}
$\infer[{con}]{\Sigma,\Gamma \vdash  c:\tau}{c:\tau\in\Sigma}\qquad\infer[{var}]{\Sigma,\Gamma \vdash x:\tau}{x:\tau \in \Gamma}\qquad\infer[{app}]{\Sigma,\Gamma \vdash  \left(M\ N\right):\tau_2}{\Sigma,\Gamma \vdash  M:\tau_1\rightarrow \tau_2 & \Sigma,\Gamma \vdash  N:\tau_1}$ \\
\vspace{1em}
$\infer[{abs}]{\Sigma,\Gamma \vdash \left(\lambda x\ .\ M\right) : \tau_1\rightarrow \tau_2}{\Sigma,\Gamma,x:\tau_1 \vdash  M :  \tau_2}\qquad\infer[{fp}]{\Sigma,\Gamma \vdash \left(\mathit{fix}\ \lambda x\, .\, M\right) : \tau}{\Sigma,\Gamma,x:\tau \vdash  M :  \tau}$
\end{center}

\paragraph{Definition of  a Sub-term}
\begin{enumerate}
\item  $\textit{Sub}\left(x\right)=\{x\}$ if $x\in \textit{Var}$ or $x\in \textit{Cst}$
\item $\textit{Sub}\left(M\ N\right)=\{M\ N\}\cup\textit{Sub}\left(M\right)\cup\textit{Sub}\left(N\right)$
\item  $\textit{Sub}\left(\lambda x\,.\, M\right)=\{\lambda x\,.\, M\}\cup\textit{Sub}\left(M\right)$
\item  $\mathit{Sub}\left(\mathit{fix}\ M\right)=\{\mathit{fix}\ M\}\cup\mathit{Sub}\left( M\right)$ 
\item in case we use $Z$ to denote an unknown (thus not analysable) term, we define that $\textit{Sub}\left(Z\right)=\{Z\}$
\end{enumerate}

\paragraph{Definition of Substitution}
\begin{enumerate}
\item $x\left[x:=N\right]\equiv N$
\item  $y\left[x:= N\right]\equiv y$, if $y\in\textit{Var}$ is distinct from $x$, or $y\in\textit{Cst}$
\item $\left(M\ L\right) \left[x:=N\right]\equiv \left(M \left[x:=N\right]\right)\ \left(L\left[x:=N\right]\right)$
\item $\left(\textit{fix}\ M\right) \left[x:=N\right]\equiv \textit{fix}\ \left(M \left[x:=N\right]\right)$
\item Let $M^{x\rightarrow y}$ denote the result of replacing every free occurrence of variable $x$ in lambda term $M$ by variable $y$.	$\left(\lambda y\,.\, P\right)\left[x:=N\right]\equiv  \lambda z\,.\,\left(P^{y\rightarrow z}\left[x:=N\right]\right)$, if both $\lambda y\,.\, P=_\alpha \lambda z\,.\,P^{y\rightarrow z}$ and $z\notin FV\left(N\right)$

\end{enumerate}

\paragraph{$\alpha$-Equivalence, Identity}
 The relation $\alpha\textit{-equivalence}$, expressed with symbol $=_\alpha$, is defined as: $\lambda x\, . \, M =_\alpha \lambda y\, . \, M^{x\rightarrow y}$, provided that $y$ does not occur in $M$. Moreover, for arbitrary term $M,N,L$ and variable $z$, we define that $M =_\alpha M$, and that if $M =_\alpha N$, then $N =_\alpha M$, $M\ L =_\alpha N\ L$, $L\ M =_\alpha L\ N$, $\lambda z \, . \, M  =_\alpha \lambda z \, . \, N$ and $\textit{fix}\ M =_\alpha \textit{fix}\ N$, and that if both $M =_\alpha N$ and $N =_\alpha L$, then $M =_\alpha L$. We use $\equiv$ to denote the relation of \emph{syntactical identity modulo $\alpha$-equivalence}.
 
 \paragraph{Syntactic Conventions}
 The outermost parentheses for a term can be omitted. Application associates to the left.
 Application binds more tightly than abstraction, therefore $\lambda x\, .\, M\ N $ stands for $\lambda x\, . \left(M\ N\right)$. The constants $\land, \lor, \supset$ are used as infix operators with precedence decreasing in the same order therein, and they all bind less tightly than application but more tightly than abstraction. For instance,  $\lambda x\ .\ p\ x\supset q\ x $ stands for $\lambda x\ .\left( \left(p\ x\right)\supset \left(q\ x\right)\right) $.  We may combine successive abstraction under one $\lambda$, for instance, $\lambda xy\,.\,M$ instead of $\lambda x . \lambda y\,.\, M$.  Successive quantification with the same quantifier can be combined under a single quantifier, for instance, we write $\exists_\iota xy\ M$ instead of $\exists_\iota x\exists_\iota y \ M$. 
\input{fp}
\section{Tree Terms and Atoms}\label{sec:trees}
 
We write $\omega$ for the set of all  non-negative integers, and write $\omega^*$ for the set of all finite lists of non-negative integers. Lists are denoted by $\left[i,\ldots,j\right]$ where $i,\ldots,j\in\omega$. The empty list is denoted $\epsilon$. If $w,v\in\omega^*$, then $\left[w,v\right]$ denotes the list which is the concatenation of $w$ and $v$. If $w\in\omega^*$ and $i \in \omega$, then $\left[w,i\right]$ denotes the list $\left[w,\left[i\right]\right]$. 

\begin{definition}[Tree Language]\label{df:tl}
A set $L \subseteq \omega^*$ is a \emph{(finitely branching)
	tree language} provided: i) for all $w \in \omega^*$ and all $i,j \in \omega$, if $\left[w,j\right] \in L$ then $w \in L$ and, for all $i<j$, $\left[w,i\right] \in L$; and ii) for all $w \in L$, the set of all $i\in \omega$ such that $\left[w,i\right]\in L$ is finite.  A non-empty tree language always contains $\epsilon$,
which we call its {\em root}. A tree language is {\em finite} if it is a finite subset of $\omega^*$, and {\em infinite} otherwise. We use $|w|$ to denote the length of list $w$, called the \emph{depth} of the node $w\in L$ if $L$ is a tree language.	
\end{definition} 

\paragraph{Arity} Given a first-order signature $\Sigma$, the type of a non-logical constant $c$ in $\Sigma$ can be depicted as $\iota\rightarrow\cdots\rightarrow\iota\rightarrow\tau$ where $\tau$ is either $\iota$ or $o$, and the \emph{arity} of $c$ is defined as the number of occurrences of $\rightarrow$ in the type of $c$. A variable of type $\iota$ has arity $0$. 

\begin{definition}[Tree-Terms, Tree-Atoms]\label{df:tt}
If $L$ is a
non-empty tree language, $\Sigma$ is a first-order signature, and   $\Gamma$ is  a context that assigns the type $\iota$ to variables, then a \emph{tree-term} over $\Sigma$ is a function $t: L \rightarrow \Sigma \cup
\Gamma$ such that, i) for all $w\in L$, $t\left(w\right)$ is either a variable in $\Gamma$, or a  non-logical constant in $\Sigma$ which is not a predicate,  and ii) the arity of  $t\left(w\right)$ equals to the cardinality of the set $\{i\in\Nat \mid \left[w,i\right] \in L\}$. If, for $t(w)$, $\{i\in\Nat \mid \left[w,i\right] \in L\} = \emptyset$, we say $t(w)$ is a \emph{leaf} of the tree-term $t$. A \emph{tree-atom} over $\Sigma$ is defined in the same way as a tree-term except that the root of $L$ is mapped to a predicate in $\Sigma$. We use $s,t$ to denote arbitrary tree-terms (tree-atoms). The domain  of a tree-term (tree-atom) $t$ is denoted $\mathit{Dom}\left(t\right)$.
A tree-term (tree-atom) $t$ is \emph{closed} if no $w\in \mathit{Dom}\left(t\right)$ is mapped to a variable, otherwise $t$ is \emph{open}.
\end{definition}

 Tree-terms (tree-atom) are finite or infinite if their domains are finite or infinite.  \emph{Substitution} of a tree-term  $s$ for all occurrences of a variable $x$ in a tree-term (tree-atom) $t$, denoted $t\left[x:=s\right]$, is defined as follows: let $t'$ be the result of the substitution, then i) $\mathit{Dom}\left(t'\right)$  is the union of $\mathit{Dom}\left(t\right)$ with the set $\{\left[w,v\right]\mid w\in \mathit{Dom}\left(t\right), v\in\mathit{Dom}\left(s\right), t\left(w\right)=x \}$, and ii) 
$t'\left(w\right)=t\left(w\right)$ if $w\in\mathit{Dom}\left(t\right)$ and $t\left(w\right)\neq x$; $t'\left(\left[w,v\right]\right)=s\left(v\right)$ if $w\in\mathit{Dom}\left(t\right), v\in\mathit{Dom}\left(s\right),\text{ and } t\left(w\right)=x$.   

\begin{example}[Tree-term]\label{ex:term-tree}
	Let 
	$\Sigma$ contain non-logical constants $0:\iota$, $\mathit{nil}:\iota$ and  $\mathit{scons}:\iota\rightarrow\iota\rightarrow\iota$. Let $\Gamma_1$ contain $x:\iota$. 
	\begin{enumerate}[leftmargin=0pt]
		\item Let $L = \{ \epsilon, \left[0\right],\left[1\right] \}$. A finite closed tree-term is defined by the 
		mapping $t_1:L\rightarrow\Sigma\cup\Gamma_1$, such that $t_1\left(\epsilon\right) = \mathit{scons}$, $t_1\left(\left[0\right]\right) = 0$, $t_1\left(\left[1\right]\right) = \mathit{nil}$.
		\item Let $W_1::= \epsilon \mid \left[W_1,1\right]$, denoting the set of all finite (possibly empty) lists of 1's. Let $L'$ be the smallest tree language containing $W_1$, which amounts to the union of $W_1$ and $W_2$, where $W_2::= \left[W_1,0\right]$.
		An infinite open tree-term is defined by the 
		mapping $t_2:L'\rightarrow\Sigma\cup\Gamma_1$, such that for all $w\in W_1$, $t_2\left(w\right)=\mathit{scons}$, and for all $v\in W_2$, $t_2\left(v\right)=\mathit{x}$ .
	\end{enumerate}
\end{example}

\begin{definition}[Tree-Term (Tree-Atom) Metric \protect\cite{Llo87}]\label{df:trunc}
Given a term $t$ on $\Sigma$ where  $\star:\iota\notin \Sigma$,	the \emph{truncation} of a tree-term (or tree-atom) $t$ at depth $n\in\omega$, denoted by $\gamma\: '(n,t)$, is constructed as follows:\\ 
	(a) the domain $\mathit{Dom}\left(\gamma\:'(n,t)\right)$ of the term $\gamma\:'(n,t)$ is $\{m \in \mathit{Dom}(t) \mid  |m|\leq n \}$; \\
	(b) 
	\[ \gamma\:'(n,t)\ (m) =
	\begin{cases}
	t(m)       & \quad \text{if } |m| < n\\
	\star   & \quad \text{if } |m| = n\\
	\end{cases}
	\]	
	\noindent	
	For tree-terms (or tree atoms) $t,s$, we define $\gamma(s,t) =
	min\{n \;| \ \gamma\:'(n,s) \neq \gamma\:'(n,t)\}$, so that $\gamma(s,t)$
	is the least depth at which $t$ and $s$ differ.  If we further define
	$d(s,t) = 0$ if $s = t$ and $d(s,t) = 2^{-\gamma(s,t)}$ otherwise,
	then the set of tree-terms and atoms on $\Sigma$ equipped with metric $d$ is an ultrametric space.
\end{definition}

\section{Lattice-Theory and Fixed Point Theorem}\label{sec:lattice}

Let set $S$ be a set equipped with a partial order $\leq$ , denoted $\langle S,\leq \rangle$, and let $S'\subseteq S$, $a\in S$ and $b\in S$.  If for all $x\in S'$, $a\leq x$, then $a$ is called a \emph{lower bound} of $S'$. If for all $x\in S'$, $x\leq b$, then $b$ is called a \emph{upper bound} of $S'$. Further, let $a,b$ be a lower bound and a upper bound of $S'$, respectively, then, if for all lower bounds $a'$ of $S'$, $a'\leq a$, then $a$ is the \emph{greatest lower bound} of $S'$, denoted ${glb}\left(S'\right)$, and, if for all upper bounds $b'$ of $S'$, $b\leq b'$, then $b$ is the \emph{least upper bound} of $S'$, denoted ${lub}\left(S'\right)$.  A partially ordered set $\langle S,\leq \rangle$  is a \emph{complete lattice} if for all subset $S'\subseteq S$, there exist ${glb}\left(S'\right)$ and ${lub}\left(S'\right)$.  
	 
	 \begin{example}\label{exmp: tarski}
	 	Given a set $S$, its power set  is denoted $\mathit{Pow}\left(S\right)$. Then $\langle\mathit{Pow}\left(S\right),\subseteq\rangle$ is a complete lattice, as for each subset $X$ of $\mathit{Pow}\left(S\right)$, ${glb}\left(X\right)$ is given by $\mybigcap \,X$,  and ${lub}\left(X\right)$ is given by $\mybigcup\, X$. 
	 \end{example}

\begin{definition}[Fixed points]\label{def:fp}
Given a complete lattice $\langle L,\leq \rangle$, a function $f: L\mapsto L$  is  \emph{increasing} if $f\left(x\right)\leq f\left(y\right)$ whenever $x\leq y$. Moreover, given an $x\in L$, $x$ is a \emph{fixed point} of $f$ if $x=f\left(x\right)$; $x$ is a \emph{pre-fixed point} of $f$ if $f\left(x\right)\leq x$; $x$ is a \emph{post-fixed point} of $f$ if $x\leq f\left(x\right)$.
We call $a\in L$ the \emph{least fixed point} of $f$, denoted $\mathit{lfp}\left(f\right)$, if $a$ is a fixed point of $f$, and for all fixed points $a'$ of $f$, $a\leq a'$. The \emph{greatest fixed point} of $f$, denoted $\mathit{gfp}\left(f\right)$, is defined similarly. 
\end{definition}

\noindent\textbf{Theorem (Knaster-Tarski). }\textit{Given a complete lattice ${\langle L,\leq \rangle}$ and an increasing function $f: L\mapsto L$ ,}	
\begin{alignat*}{3}
\mathit{lfp}\left(f\right)&=\mathit{glb}\{x\mid x=f\left(x\right)\}&=\mathit{glb}\{x\mid f\left(x\right)\leq x\}\\
\mathit{gfp}\left(f\right)&=\mathit{lub}\{x\mid x=f\left(x\right)\}&=\mathit{lub}\{x\mid x \leq  f\left(x\right)\}
\end{alignat*}

\section{Full Proof of Soundness Theorems  
\ref{thm: main soundess infinite} and \ref{them: cut theorem infi model} }\label{sec:sproof}

In this Section, we give full proofs for the main Lemma and Theorems of Section~\ref{sec:sound}.

\subsection{Corollaries for Lemma~\ref{lem:prod} (Productivity Lemma)} 
Corollary \ref{col: produc} and Corollary \ref{col: guarded atom confluence} follow from Lemma~\ref{lem:prod}.

\begin{corollary}[Productivity of Guarded Full Terms ]\label{col: produc}
	If $M$ (which is not a first order term) on $\Sigma$ is a closed guarded full term, or $M$ is $\mathit{fix}\beta$-equivalent to a closed guarded full term, then there exist an equivalent infinite tree-term $M^T$. Further, if  $M$ is closed, then $M^T\in\mathcal{H}^\Sigma$.  
\end{corollary}

\begin{corollary}[Tree Sharing Among \protect\textit{fix}$\beta$-Equivalent Guarded Atoms]
\label{col: guarded atom confluence}
If $A_1,\ldots, A_n$ are guarded atoms and they are pairwise $\mathit{fix}\beta$-equivalent, then $A^T_1=\ldots=A_n^T$.
	\end{corollary}

\subsection{Lemma~\protect\ref{lem: coin prf pcpl infi} (Coinductive proof principle)}

\textbf{Lemma~\protect\ref{lem: coin prf pcpl infi} (Coinductive proof principle for coinductive models).} \textit{Let $P_{\Sigma}$ be a logic program, with the operator $\mathcal{T}$ and 
	the model $\mathcal{M}$. A set 
	$S\subseteq\mathcal{M}$, if and only if,  there exist a Herbrand interpretation  $I$ for $P_{\Sigma}$, such that $S\subseteq I$ and $I$ is a post-fixed point of $\mathcal{T}$.  }

\begin{proof}
We first prove the implication from right to left. This is justified by the definition of $\mathcal{M}$.

Then we prove from left to right. $S\subseteq\mathcal{M}$ implies that for all $x\in S$, $x \in \mathcal{M} $. Then, by the definition of $\mathcal{M}$, for all $x\in S$, there exist an $I_x$, such that $x\in I_x$ and $I_x\subseteq\mathcal{T}\left(I_x\right)$. Then let $I=\bigcup\{I_x\mid  x\in S\}$. By the construction of $I$,  we have  $S\subseteq I$. Also we have $I_x\subseteq I$, for all $I_x$. Since $\mathcal{T}$ is increasing, $\mathcal{T}\left(I_x\right)\subseteq\mathcal{T}\left(I\right)$, for all $I_x$. By transitivity of $\subseteq$, $I_x\subseteq\mathcal{T}\left(I\right)$, for all $I_x$. So $I\subseteq\mathcal{T}\left(I\right)$, by the construction of $I$ . 
\end{proof}

\subsection{Proof of Theorem~\protect\ref{thm: main soundess infinite} (Soundness of Coinductive Uniform Proofs)}

\vspace{9pt}
\noindent\textbf{Theorem~\protect\ref{thm: main soundess infinite} (Soundness of \textit{co-hohh} proofs).} \textit{Let
$P_{\Sigma}$ be a logic program with coinductive model $\mathcal{M}$,
${\forall\!_{\iota}x_1\ldots x_m\ A_1\land \ldots \land A_n \supset A}$ be a $H$-formula, and 
$\Sigma;P\looparrowright \forall\!_{\iota}x_1\ldots x_m\ A_1\land \ldots \land A_n \supset A$ have a  proof in \textit{co-hohh} which involves only guarded atoms. 
Then, for an arbitrary tree-form ground instance ${A'_1\land \ldots \land A'_n \supset A'}$ of the goal, $\{A'_1,\ldots A'_n\}\subseteq\mathcal{M}$ implies $\{A'\}\subseteq\mathcal{M}$.}

\begin{proof}
~
\begin{enumerate}[leftmargin=0pt]
	\item \label{enum item: requirements intro}We use Lemma~\ref{lem: coin prf pcpl infi} from right to left, which means, to show that $\{A'\}\subseteq \mathcal{M}$, we look for a set  $I$ such that the \emph{requirements}
	${\{A'\}\subseteq I}$ and ${I \subseteq \mathcal{T}\left(I\right)}$ are satisfied.
\item	\label{enum item: proof structure intro}The proof follows an Analysis--Construction--Verification structure, where we first study the proof of the \emph{root sequent} \[\Sigma;P\looparrowright \forall\!_{\iota}x_1\ldots x_m\ A_1\land \ldots \land A_n \supset A\] which provides a clue for constructing a candidate set $I$. Finally we verify that the set $I$ so constructed satisfies the \emph{requirements}.
\end{enumerate}	
	
	\noindent\textbf{Analysis} 
\begin{enumerate}[resume,leftmargin=0pt]
	\item \label{enum item: analysis co-fix and for all right}The proof for the root sequent starts with the following steps.
	\[\infer[\textsc{co-fix}]
	{\Sigma;P\looparrowright \forall\!_{\iota}x_1\ldots x_m\ A_1\land \ldots \land A_n \supset A}
	{
		\infer[\forall R\langle\rangle ]
		{\Sigma;P,\mathit{ch}\longrightarrow \langle\forall\!_{\iota}x_1\ldots x_m\ A_1\land \ldots \land A_n \supset A\rangle}
		{
			\infer[\forall R\langle\rangle ]
			{\vdots\quad \textit{(Using only $\forall R\langle\rangle$ until) }}
			{\overline{c:\iota\vphantom{d}},\Sigma;P,\mathit{ch}\longrightarrow \langle \left(A_1\land \ldots \land A_n \supset A\right)\overline{\left[x:=c\right]\vphantom{\bar{\left[\right]}}}\rangle}
		} 
	}
	\]
	where $\mathit{ch}$ (the coinductive hypothesis) is
	\[\forall\!_{\iota}x_1\ldots x_m\ A_1\land \ldots \land A_n \supset A\] and $\overline{c:\iota\vphantom{d}}$ is a short hand for $c_1:\iota,\ldots,c_m:\iota$, and  $\overline{\left[x:=c\right]\vphantom{\bar{\left[\right]}}}$ is a short hand for $\left[x_1:=c_1\right]\ldots\left[x_m:=c_m\right]$.
	\item \label{enum item: signature extension and referent defined}The $\forall R\langle\rangle$ steps in \ref{enum item: analysis co-fix and for all right} extend the signature from $\Sigma$ to $\overline{c:\iota\vphantom{d}},\Sigma$.
	Given a formula $K$, we will distinguish its term-form ground instance based on $\Sigma$ (denoted with $\lfloor K\rfloor$) from its term-form ground instance based on ${\overline{c:\iota\vphantom{d}},\Sigma}$ (denoted with $\lfloor K \rfloor^c$).
	Moreover, if a lambda term $L$ on ${\overline{c:\iota\vphantom{d}},\Sigma}$,
	contains some (or none) of the eigenvariables $c_1, \ldots,c_m$,
	we will use notation $L^c$ to refer to it, and particularly, we reserve $B^c$ for an arbitrary such atomic $H$-formula. We define a \emph{referent} of $(B^c)^T$ (provided $B^c$ is guarded),
	as \[(B^c)^T\left[c_{1}:= N_1\right]\ldots\left[c_{m}:= N_m\right]\] where $N_k\in\mathcal{H}^\Sigma$, $1\leq k\leq m$, 
	and, regarding the definition of substitution,  $c_{1}, \ldots,c_{m}$ are technically treated as free variables in $(B^c)^T$. 	
	\item \label{enum item: add hp and decide}
	The proof for the root sequent proceeds as  follows:
	\[
	\infer[\supset\!\! R\langle\rangle]
	{\overline{c:\iota\vphantom{d}},\Sigma;P,\mathit{ch}\longrightarrow \langle \left(A_1\land \ldots \land A_n \supset A\right)\overline{\left[x:=c\right]\vphantom{\bar{\left[\right]}}}\rangle}
	{
		\infer[\textsc{decide}\langle\rangle]
		{
			\overline{c:\iota\vphantom{d}},\Sigma;P,\mathit{ch},\mathit{hp}\longrightarrow \langle  A\overline{\left[x:=c\right]\vphantom{\bar{\left[\right]}}}\rangle
		}
		{
			\overline{c:\iota\vphantom{d}},\Sigma;P,\mathit{ch},\mathit{hp}\stackrel{D^*}{\longrightarrow} \langle  A\overline{\left[x:=c\right]\vphantom{\bar{\left[\right]}}}\rangle
		}
	}
	\]
	where $\mathit{hp}$ (the hypothesis) is $\left(A_1\land \ldots \land A_n\right) \overline{\left[x:=c\right]\vphantom{\bar{\left[\right]}}}$ and ${D^*\in P}$. Note that the above $\textsc{decide}\langle\rangle$ step is also the only occasion in the proof where the $\textsc{decide}\langle\rangle$ rule is used (and this shall become clearer as we further analyse the sequent proof).
	We will use $\forall\!_\iota \bar{x}\ A_{\alpha}$ as the general form of a \emph{fact} in $P$,  ${\forall\!_\iota \bar{x}\;A_{\alpha_1}\land\ldots\land A_{\alpha_u}\supset A_{\alpha_0}}$ as the general form of a \emph{rule} in $P$, and write ${A'_{\alpha_1}\land\ldots\land A'_{\alpha_u}\supset A'_{\alpha_0}}$ for ${\lfloor\forall\!_\iota \bar{x}\;A_{\alpha_1}\land\ldots\land A_{\alpha_u}\supset A_{\alpha_0}\rfloor^c}$.
	$D^*$ can be either a fact or a rule.
	\item \label{enum item: D star is a fact} 	If $D^*$ is a fact, then the proof for the root sequent proceeds as  follows:
	\[
	\infer[\forall L\langle\rangle]
	{\overline{c:\iota\vphantom{d}},\Sigma;P,\mathit{ch},\mathit{hp}\stackrel{\forall\!_\iota \bar{x}\ A_{\alpha} }{\longrightarrow} \langle  A\overline{\left[x:=c\right]\vphantom{\bar{\left[\right]}}}\rangle}
	{  \infer[\forall L\langle\rangle]
		{\vdots\quad \textit{(Using only $\forall L\langle\rangle$ until) }}
		{
			\infer[\textsc{initial}\langle\rangle]
			{
				\overline{c:\iota\vphantom{d}},\Sigma;P,\mathit{ch},\mathit{hp}\stackrel{\lfloor\forall\!_\iota \bar{x}\ A_{\alpha} \rfloor^c}{\longrightarrow}\langle  A\overline{\left[x:=c\right]\vphantom{\bar{\left[\right]}}}\rangle
			}{} 
		}    
	} 
	\] 
	where $\lfloor\forall\!_\iota \bar{x}\ A_{\alpha} \rfloor^c =_{\mathit{fix}\beta} A\overline{\left[x:=c\right]\vphantom{\bar{\left[\right]}}}$. This situation is trivial and we can do the post-fixed point construction and verification immediately, as follows. Note that $A'$ (which the head of an arbitrary tree-form ground instance of the goal mentioned in the theorem) is a referent of $\left(A\overline{\left[x:=c\right]\vphantom{\bar{\left[\right]}}}\right)^T$. Since $\lfloor\forall\!_\iota \bar{x}\ A_{\alpha} \rfloor^c =_{\mathit{fix}\beta} A\overline{\left[x:=c\right]\vphantom{\bar{\left[\right]}}}$, then by Corollary~\ref{col: guarded atom confluence},  $\left(A\overline{\left[x:=c\right]\vphantom{\bar{\left[\right]}}}\right)^T= \left(\lfloor\forall\!_\iota \bar{x}\ A_{\alpha} \rfloor^c\right)^T$, so $A'$ is a referent of $\left(\lfloor\forall\!_\iota \bar{x}\ A_{\alpha} \rfloor^c\right)^T$. A referent of $\left(\lfloor\forall\!_\iota \bar{x}\ A_{\alpha} \rfloor^c\right)^T$, i.e. $A'$, is a tree-form ground instance of $\forall\!_\iota \bar{x}\ A_{\alpha}$. Then, let $I=\{A'\}$, we have  ${I \subseteq \mathcal{T}\left(I\right)}$. Now we have verified that $I$ satisfies the requirements (cf. \ref{enum item: requirements intro}).
\item \label{enum item: D star is a rule}		If $D^*$ is a rule, then the proof for the root sequent proceeds as  follows:
\[
\infer[\forall L\langle\rangle]
{\overline{c:\iota\vphantom{d}},\Sigma;P,\mathit{ch},\mathit{hp}\stackrel{\forall\!_\iota \bar{x}\;A_{\alpha_1}\land\ldots\land A_{\alpha_u}\supset A_{\alpha_0}}{\longrightarrow} \langle  A\overline{\left[x:=c\right]\vphantom{\bar{\left[\right]}}}\rangle}
{  \infer[\forall L\langle\rangle]
	{\vdots\quad \textit{(Using only $\forall L\langle\rangle$ until) }}
	{
		\infer[\supset\! L\langle\rangle]
		{
			\overline{c:\iota\vphantom{d}},\Sigma;P,\mathit{ch},\mathit{hp}\stackrel{\lfloor\forall\!_\iota \bar{x}\;A_{\alpha_1}\land\ldots\land A_{\alpha_u}\supset A_{\alpha_0}\rfloor^c}{\longrightarrow}\langle  A\overline{\left[x:=c\right]\vphantom{\bar{\left[\right]}}}\rangle 
		}
		{\mathit{SubProof\ 1 } & \mathit{SubProof\ 2 }} 
	}   
} 
\]
where $\mathit{SubProof\ 1 }$ is 
\[
\infer[\textsc{initial}]
{
	\overline{c:\iota\vphantom{d}},\Sigma;P,\mathit{ch},\mathit{hp}\stackrel{A'_{\alpha_0}}{\longrightarrow}  A\overline{\left[x:=c\right]\vphantom{\bar{\left[\right]}}}
}{}
\]
and the  root of $\mathit{SubProof\ 2 }$ is labelled by
\[\overline{c:\iota\vphantom{d}},\Sigma;P,\mathit{ch},\mathit{hp}{\longrightarrow} A'_{\alpha_1}\land\ldots\land A'_{\alpha_u}\] 

\item \label{enum item: remark on impli right angle step} The $\supset\! R\langle\rangle$ step in \ref{enum item: D star is a rule} removes the coinductive status of goals in its upper sequents, then $\mathit{SubProof\ 2 }$ is a uniform proof in the style of Miller and Nadathur \cite{MNPS91}. It is important to know that $\mathit{SubProof\ 2 }$ has the following properties.

\item\label{enum item: hh three forms of sequents in proof}  We pay attention to two kinds of sequents involved in $\mathit{SubProof\ 2 }$:
		\begin{enumerate}[leftmargin=*]
			\item\label{enum item: hh Ac1 to Acp}  Sequents, whose right side is a conjunction of (at least two) atoms,  in the form
			\begin{flushleft}
				$\overline{c:\iota\vphantom{d}},\Sigma;P,\mathit{ch},\mathit{hp}\longrightarrow B^c_{1}\land\ldots \land B^c_{p}$ 
			\end{flushleft} 
			For instance, the sequent labelling the root of $\mathit{SubProof\ 2 }$ is in this form.
			\item \label{enum item: hh Ac} Sequents, whose right side is a single atom, in the form
			\begin{flushleft}
				$\overline{c:\iota\vphantom{d}},\Sigma;P,\mathit{ch},\mathit{hp}\longrightarrow B^c$ 
			\end{flushleft}
			
		\end{enumerate}
		\item \label{enum item: hh process Ac1 Acp}In $\mathit{SubProof\ 2 }$, for each sequent in the form (\ref{enum item: hh Ac1 to Acp}), a succession of $\land R$ steps are used to break it down to a series of sequents \[\overline{c:\iota\vphantom{d}},\Sigma;P,\mathit{ch},\mathit{hp}\longrightarrow B^c_{k}\quad 1\leq k\leq p\] which are in the form (\ref{enum item: hh Ac}).
		
		\item\label{enum item: hh process Ac} In $\mathit{SubProof\ 2 }$, for each sequent in the form (\ref{enum item: hh Ac}), only the \textsc{decide} rule can be used as the next step of reduction, and there are four cases of possible development, as follows.
		
		\begin{enumerate}
			\item \label{enum item: hh quan imp decide axiom}\textsc{decide} selects a fact from $P$.
			\[
			\infer[\textsc{decide}]
			{\overline{c:\iota\vphantom{d}},\Sigma;P,\mathit{ch},\mathit{hp}{\longrightarrow}  B^c}
			{
				\infer[\forall L\langle\rangle]
				{\overline{c:\iota\vphantom{d}},\Sigma;P,\mathit{ch},\mathit{hp}\stackrel{\forall\!_\iota \bar{x}\ A_{\alpha}}{\longrightarrow}  B^c}
				{  \infer[\forall L\langle\rangle]
					{\vdots\quad \textit{(Using only $\forall L\langle\rangle$ until) }}
					{
						\infer[\textsc{initial}]
						{
							\overline{c:\iota\vphantom{d}},\Sigma;P,\mathit{ch},\mathit{hp}\stackrel{\lfloor\forall\!_\iota \bar{x}\ A_{\alpha} \rfloor^c}{\longrightarrow} B^c
						}{} 
					}    
				} 
			}
			\] 
			where $B^c =_{\mathit{fix}\beta} {\lfloor\forall\!_\iota \bar{x}\ A_{\alpha} \rfloor^c}$.

			\item \label{enum item: hh quan imp decide non-axiom} \textsc{decide} selects a rule from $P$.
			\[
			\infer[\textsc{decide}]
			{\overline{c:\iota\vphantom{d}},\Sigma;P,\mathit{ch},\mathit{hp}{\longrightarrow}  B^c}
			{
				\infer[\forall L\langle\rangle]
				{\overline{c:\iota\vphantom{d}},\Sigma;P,\mathit{ch},\mathit{hp}\stackrel{\forall\!_\iota \bar{x}\;A_{\alpha_1}\land\ldots\land A_{\alpha_u}\supset A_{\alpha_0}}{\longrightarrow}  B^c}
				{  \infer[\forall L\langle\rangle]
					{\vdots\quad \textit{(Using only $\forall L\langle\rangle$ until) }}
					{
						\infer[\supset\! L\langle\rangle]
						{
							\overline{c:\iota\vphantom{d}},\Sigma;P,\mathit{ch},\mathit{hp}\stackrel{\lfloor\forall\!_\iota \bar{x}\;A_{\alpha_1}\land\ldots\land A_{\alpha_u}\supset A_{\alpha_0}\rfloor^c}{\longrightarrow}  B^c 
						}
						{\mathit{SubProof\ 1' } & \mathit{SubProof\ 2' }} 
					}   
				}
			} 
			\]
			where $\mathit{SubProof\ 1' }$ is 
			\[
			\infer[\textsc{initial}]
			{
				\overline{c:\iota\vphantom{d}},\Sigma;P,\mathit{ch},\mathit{hp}\stackrel{A'_{\alpha_0}}{\longrightarrow} B^c
			}{}
			\]
			and the  root of $\mathit{SubProof\ 2' }$ is labelled by
			\[\overline{c:\iota\vphantom{d}},\Sigma;P,\mathit{ch},\mathit{hp}{\longrightarrow} A'_{\alpha_1}\land\ldots\land A'_{\alpha_u}\] 
			which is a sequent of the form (\ref{enum item: hh Ac1 to Acp}). Note that the selected rule is possibly different from the rule selected in \ref{enum item: D star is a rule} although the same formula is used to depict the selected rule. 
			
			\item \label{enum item: hh quan impl decide CH}\textsc{decide} selects the coinductive hypothesis $\mathit{ch}$.
			\[\infer[\textsc{decide}]
			{\overline{c:\iota\vphantom{d}},\Sigma;P,\mathit{ch},\mathit{hp}\longrightarrow B^c}
			{
				\infer[\forall L]
				{\overline{c:\iota\vphantom{d}},\Sigma;P,\mathit{ch},\mathit{hp}\stackrel{\forall\!_\iota x_1\ldots x_m\ A_1\land \ldots \land A_n\supset A}{\longrightarrow} B^c}
				{
					\infer[\forall L]
					{\vdots\quad \textit{(Using only $\forall L$ until) }}
					{
						\infer[\supset\!\! L]
						{
							\overline{c:\iota\vphantom{d}},\Sigma;P,\mathit{ch},\mathit{hp}\stackrel{\lfloor\forall\!_\iota x_1\ldots x_m\ A_1\land \ldots \land A_n\supset A\rfloor^c}{\longrightarrow} B^c
						}
						{\mathit{SubProof\ 1'' } & \mathit{SubProof\ 2'' }}
					}
				}
			}\]
			where $\lfloor\forall\!_\iota x_1\ldots x_m\ A_1\land \ldots \land A_n\supset A\rfloor^c$ refers to \[\left(A_1\land \ldots \land A_n\supset A\right)\left[x_{1}:= L^c_1\right]\ldots\left[x_{m}:= L^c_m\right]\] for some guarded full terms $L^c_1\ldots L^c_m\in\mathcal{U}^\Sigma_2$--- These terms are guarded full terms because the theorem assumes that all atoms involved in the sequent proof are guarded, both before substitution and after substitution. Then a case analysis on a variable being substituted, reveals that the substituting terms must be guarded full terms. 
			 $\mathit{SubProof\ 1'' }$ is 
			\[
			\infer[\textsc{initial}]{\overline{c:\iota\vphantom{d}},\Sigma;P,\mathit{ch},\mathit{hp}\stackrel{A\left[x_{1}:= L^c_1\right]\ldots\left[x_{m}:= L^c_m\right]}{\longrightarrow} B^c}{}
			\]
			and the root of $\mathit{SubProof\ 2'' }$ is labelled by
			\[{\overline{c:\iota\vphantom{d}},\Sigma;P,\mathit{ch},\mathit{hp}{\longrightarrow} \left(A_1\land \ldots \land A_n\right)\left[x_{1}:= L^c_1\right]\ldots\left[x_{m}:= L^c_m\right]}\]
			which is a sequent of the form (\ref{enum item: hh Ac1 to Acp}). Let $\delta$ denote the substitution  \[\left[x_{1}:= L^c_1\right]\ldots\left[x_{m}:= L^c_m\right]\] Later we will see that the above mentioned $\delta$ plays an central role in the construction of the post-fixed point.

			\item \label{enum item: hh quan impl decide HP}\textsc{decide} selects the hypothesis $\mathit{hp}$. 
			\[\infer[\textsc{decide}]{\Sigma;P,\mathit{ch},\mathit{hp}\longrightarrow B^c}{
				\infer[\land L]
				{\overline{c:\iota\vphantom{d}},\Sigma;P,\mathit{ch},\mathit{hp}\stackrel{\left(A_1\land \ldots \land A_n\right)\overline{\left[x:=c\right]\vphantom{\bar{\left[\right]}}}}{\longrightarrow} B^c}
				{
					\infer[\land L]
					{\vdots\quad \textit{(Using only $\land L$ until) }}
					{\infer[\textsc{initial}]{\overline{c:\iota\vphantom{d}},\Sigma;P,\mathit{ch},\mathit{hp}\stackrel{A_k\overline{\left[x:=c\right]\vphantom{\bar{\left[\right]}}}}{\longrightarrow} B^c}{}}
				}
			}\]
			where $B^c =_{\mathit{fix}\beta} A_k\overline{\left[x:=c\right]\vphantom{\bar{\left[\right]}}}$ for some $k$ and ${1\leq k\leq n}$.
		\end{enumerate}  		
	\end{enumerate}
	
	\noindent\textbf{Construction}
\begin{enumerate}[leftmargin=0pt, resume]
	\item \label{enum item: construction intro}
 We now use the above analysis of the proof of the root sequent to construct the post-fixed point $I$.
 It remains to work with the case when the program clause $D^*$ selected by $\textsc{decide}\langle\rangle$ is a rule (cf. \ref{enum item: add hp and decide}, \ref{enum item: D star is a fact}, \ref{enum item: D star is a rule}).
 The main difficulty in construction in this situation arises when the post-fixed point in question is an infinite set, which 
 happens when a proof exhibits certain irregularity.
 
 \item \label{enum item: provide concrete terms for instan the goal} Firstly, let the formula ${A'_1\land \ldots \land A'_n \supset A'}$ (as mentioned in the theorem) be computed by \[\left(A^T_1\land \ldots \land A^T_n\supset A^T\right)\left[x_{1}:= \mathcal{N}_1\right]\ldots\left[x_{m}:= \mathcal{N}_m\right]\] for some arbitrary  $\mathcal{N}_1,\ldots, \mathcal{N}_m\in \mathcal{H}^\Sigma$. 
 
 \item \label{enum item: assume ch is used for s times}Note that the rules of our coinductive calculus allow us to form only one coinductive hypothesis in a proof for the root sequent, but the proof may use this coinductive hypothesis more than once. 
 Suppose that the coinductive hypothesis $\mathit{ch}$ is used in the given proof $s\in\omega$ times. 
 Let $\psi=\{1,\ldots,s\}$.
 
 \item \label{enum item: substi deltas denoted}Consider the $s$ times when the coinductive hypothesis was used. In each case, a substitution $\delta$ was constructed, as described above in \ref{enum item: hh quan impl decide CH}.
 Let us denote all these substitutions by $\delta_1,\ldots,\delta_s$.
 
 \item Each $\delta_j$ ($1\leq j\leq s$) is given by \[\left[x_{1}:= \mathcal{L}^c_{\left(j,\,1\right)}\right]\ldots\left[x_{m}:= \mathcal{L}^c_{\left(j,\,m\right)}\right]\]	
 where $\mathcal{L}^c_{\left(j,\,i\right)}$ is given by $L^c_{i}$ in $\delta_j$, as in  \ref{enum item: hh quan impl decide CH}.
 (Note that each time the same coinductive hypothesis is used.)
 
 \item Given that the proof of the root sequent uses exactly the eigenvariables $c_1, \ldots c_m$,  we want to define a set of mapping for the eigenvariables into $\mathcal{H}^\Sigma$ in order to construct the post-fixed point.  So, we define $\psi^*$ to denote the set of all finite lists on $\psi$, including the empty list $\epsilon$, and define
 the set $ \{\Theta\left( w \right) \mid w \in \psi^*\}$ of substitutions, where 
 $\Theta$ is a function defined by recursion, as follows. 
 \begin{enumerate}
 	\item Base case:\[\Theta\left(\epsilon\right)  =\left[c_{1}:= \mathcal{N}_1\right]\ldots\left[c_{m}:= \mathcal{N}_m\right]\]
 	\item Recursive case: If  $w\in\psi^*,\ j\in\psi$, then\[\Theta\left(\left[w,j\right]\right)  =\left[c_{1}:= \left( \left(\mathcal{L}^c_{\left(j,1\right)}\right)^T\Theta\left(w\right) \right)\right]\ldots\left[c_{m}:= \left( \left(\mathcal{L}^c_{\left(j,m\right)}\right)^T\Theta\left(w\right) \right)\right]\]
 	where the notation $\left(\mathcal{L}^c_{\left(j,m\right)}\right)^T\Theta\left(w\right)$ is used to denote application of the substitution $\Theta\left(w\right)$ to the tree-term $\left(\mathcal{L}^c_{\left(j,m\right)}\right)^T$. Note that  $\left(\mathcal{L}^c_{\left(j,m\right)}\right)^T$ is a valid notation due to Corollary~\ref{col: produc} and the fact that $\mathcal{L}^c_{\left(j,m\right)}$ is a guarded full term (cf. \ref{enum item: hh quan impl decide CH}).
 \end{enumerate}

 Since $\psi^*$ is an infinite set, this construction is potentially infinite, and moreover, through recursive application of substitutions at each iteration of $\Theta$ it can construct an infinite set of substitutions. 
 
 \item Note that the above construction involving $\Theta$ can and will be used to construct models generally, even when the coinductive hypothesis was not applied (i.e applied $s=0$ times). 
 
 \item We define the set $I^c$, as the collection of all atoms $\left(E^c\right)^T$, where $E^c$ is the atom on the right side of some sequent $Q$, where $Q$ is in $\mathit{SubProof\ 2 }$.
 In addition, we add $\left(A\overline{\left[x:=c\right]\vphantom{\bar{\left[\right]}}}\right)^T$ to $I^c$, i.e. we add the tree corresponding to the atom occurring in the lower sequent in the $\textsc{decide}\langle\rangle$ step (cf. \ref{enum item: add hp and decide}).
 Note that the  $\langle A\overline{\left[x:=c\right]\vphantom{\bar{\left[\right]}}} \rangle$ in that sequent should intuitively be seen as a coinductive conclusion.
 
 \item \label{enum item: define IC} We use $I^c\Theta\left(w\right)$ to denote the set resulted from applying substitution $\Theta\left(w\right)$ to all members of $I^c$. We define $I_1$ as \[I_1={\bigcup_{w\in\psi^*}} \; I^c\Theta\left(w\right)\] 
 
 \item Given $\{ A'_1,\ldots A'_n\}\subseteq \mathcal{M}$, by Lemma~\ref{lem: coin prf pcpl infi} (used from left to right), there exist a set $I_2$,  such that ${\{A'_1,\ldots A'_n\}\subseteq I_2}$ and $I_{2}$ is a post-fixed point of $\mathcal{T}$. 
 
 \item Now we construct a candidate post-fixed point $I$, given as \[I=I_1 \cup I_2\] 
\end{enumerate}	
	 	
	\noindent\textbf{Verification} 
\begin{enumerate}[resume,leftmargin=0pt]
	\item We verify that set $I$ satisfies the two requirements (cf. \ref{enum item: requirements intro}). 
	
	\item \label{enum item: requi one satisfied} $\{A'\}\subseteq I$ is proved by
	\begin{align*}
	&  \{A'\}\\
	= &\  \left\lbrace \left(A\overline{\left[x:=c\right]\vphantom{\bar{\left[\right]}}}\right)^T\Theta\left(\epsilon\right)\right\rbrace \\
	\subseteq &\ I^c\Theta\left(\epsilon\right)\subseteq I_1 \subseteq I
	\end{align*}
	
	\item  It remains to show that $I$ is a post-fixed point of $\mathcal{T}$. In other words, we need to show that each member of $I$ is also a member of $\mathcal{T}\left(I\right)$. 
	
	\item For all $y\in I_2$, $y\in \mathcal{T}\left(I\right)$, as proven by the  following argument:
	
	\begin{align*}
	I_2\subseteq I&\quad\text{by def. of }I\\
	\mathcal{T}\left(I_2\right)\subseteq\mathcal{T}\left(I\right)&\quad\mathcal{T}\text{ is increasing}\\
	I_2\subseteq\mathcal{T}\left(I_2\right)&\quad\text{by def. of }I_2\\
	I_2\subseteq\mathcal{T}\left(I\right)&\quad\text{by transitivity of }\subseteq
	\end{align*}
	
	\item By definition of $I_1$, for all  $y\in I_1$,   $y\in I^c\Theta\left(w\right)$ for some $w\in\psi^*$, and $y$ has an underlying atom $\left(F^c\right)^T\in I^c$ such that  \[y=\left(F^c\right)^T\Theta\left(w\right)\] Let $R$ be the sequent from where we get the atom $F^c$. We show $y\in \mathcal{T}\left(I\right)$ using a case analysis \ref{enum item: veri casse decide fact} -- \ref{enum item: veri casse decide hp},  with respect to the ways \ref{enum item: hh quan imp decide axiom} -- \ref{enum item: hh quan impl decide HP} in which the sequent $R$ is reduced.
	\begin{enumerate}
		\item  \label{enum item: veri casse decide fact} If $R$ is reduced in the way \ref{enum item: hh quan imp decide axiom}, then there exist a fact $\forall\!_\iota \bar{x}\ A_{\alpha}$ in $P$, such that $F^c =_{\mathit{fix}\beta} \lfloor\forall\!_\iota \bar{x}\ A_{\alpha}\rfloor^c$, then $\left(F^c\right)^T = \left(\lfloor\forall\!_\iota \bar{x}\ A_{\alpha}\rfloor^c\right)^T$ (by Corollary~\ref{col: guarded atom confluence}), so there exist tree-form ground instance $\lfloor\forall\!_\iota \bar{x}\ A_{\alpha}\rfloor^T$ such that  $\left(F^c\right)^T\Theta\left(w\right)=\lfloor\forall\!_\iota \bar{x}\ A_{\alpha}\rfloor^T$. So $y\in \mathcal{T}\left(I\right)$.
		
		\item \label{enum item: veri casse decide rule}If $R$ is reduced in the way \ref{enum item: hh quan imp decide non-axiom} (note that this is also the case when $F^c\equiv A{\overline{\left[x:=c\right]\vphantom{\bar{\left[\right]}}} }$), then there exist a rule $K\in P$ such that  $\left(\mathit{head}\ \lfloor K\rfloor^c\right)=_{\mathit{fix}\beta}F^c$, and for each $x\in\left(\mathit{body}\ \lfloor K\rfloor^c\right)$, $x^T\in I^c$. Note that  $\left(\mathit{head}\ \lfloor K\rfloor^c\right)^T= \left(F^c\right)^T$ by Corollary~\ref{col: guarded atom confluence}. Then there exist a tree-form ground instance $\lfloor K\rfloor^T$ for $K$, whose head is $\left(F^c\right)^T\Theta\left(w\right)$, and whose body $S$ is the set such that $x'\in S$ if and only if there exist $x\in\left(\mathit{body}\ \lfloor K\rfloor^c\right)$ and $x'=x^T\,\Theta\left(w\right)$. Then, for all $x'\in S$, $x'\in I^c\Theta\left(w\right)$. So  $y\in\mathcal{T}\left(I\right)$.
		
		\item \label{enum item: veri casse decide ch}If $R$ is reduced in the way \ref{enum item: hh quan impl decide CH}, then the use of $ch$ results in existence of $\delta_r$ for some $r\in\psi$, given as \[\left[x_{1}:= \mathcal{L}^c_{\left(r,1\right)}\right]\ldots\left[x_{m}:= \mathcal{L}^c_{\left(r,m\right)}\right]\]and $F^c =_{\mathit{fix}\beta} A\delta_r$. Note that 
		\[\left(A\,\delta_r\right)^T\Theta\left(w\right)\ =\  \left(A\overline{\left[x:=c\right]\vphantom{\bar{\left[\right]}}}\right)^T\Theta\left(\left[w,r\right]\right)\] 
		because, by definition of $\Theta$, 
		\[\Theta\left(\left[w,r\right]\right)= \left[c_{1}:=\left( \left(\mathcal{L}^c_{\left(r,1\right)}\right)^T\Theta\left(w\right)\right)\right]\ldots\left[c_{m}:= \left(\left(\mathcal{L}^c_{\left(r,m\right)}\right)^T\Theta\left(w\right)\right)\right]\]
		So $y\in I^c\Theta\left(\left[w,r\right]\right)$ and this membership is underlain by $\left(A\overline{\left[x:=c\right]\vphantom{\bar{\left[\right]}}}\right)^T\in I^c$. In light of this discovery, we start a new round of case analysis where we regard the same $y$ as a member of $I^c\Theta\left(\left[w,r\right]\right)$, then \ref{enum item: veri casse decide rule} applies, so $y\in \mathcal{T}\left(I\right)$.
		
		\item \label{enum item: veri casse decide hp} If $R$ is reduced in the way \ref{enum item: hh quan impl decide HP}, it implies that $F^c=_{\mathit{fix}\beta}A_k\overline{\left[x:=c\right]\vphantom{\bar{\left[\right]}}}$ for some $k$. Further, by Corollary~\ref{col: guarded atom confluence}, $\left(F^c\right)^T = \left(A_k\overline{\left[x:=c\right]\vphantom{\bar{\left[\right]}}}\right)^T$. Then $\left(F^c\right)^T \Theta\left(w\right) = \left(A_k\overline{\left[x:=c\right]\vphantom{\bar{\left[\right]}}}\right)^T\Theta\left(w\right)$.  We further inspect two sub-cases.
		\begin{enumerate}
			\item \label{enum item: veri casse decide hp sub case empty} $w=\epsilon$. Then  $y=A'_k$, and $A'_k\in I_2$, so $y\in\mathcal{T}\left(I\right)$.
			\item \label{enum item: veri casse decide hp sub case nonempty}
			
			$w\neq\epsilon$. Then there exist $v\in\psi^*, i\in\psi$ such that $\left[v,i\right]=w$. Then,  there exist a use of $\mathit{ch}$ that gives rise to $\delta_i$, given as
			\[\left[x_{1}:= \mathcal{L}^c_{\left(i,1\right)}\right]\ldots\left[x_{m}:= \mathcal{L}^c_{\left(i,m\right)}\right]\]
			In addition, there exist ${\left(A_k\delta_i\right)^T\in I^c}$ (cf. \ref{enum item: hh quan impl decide CH} and \ref{enum item: define IC}). Note that 
			\[\left(A_k\,\delta_i\right)^T\Theta\left(v\right)\ =\  \left(A_k\overline{\left[x:=c\right]\vphantom{\bar{\left[\right]}}}\right)^T\Theta\left(w\right)\]
			This is because , by definition of $\Theta$,  that
			\begin{align*}
			&\Theta\left(w\right)=\Theta\left(\left[v,i\right]\right)\\
			=\ & \left[c_{1}:= \left(\left(\mathcal{L}^c_{\left(i,1\right)}\right)^T\Theta\left(v\right)\right)\right]\ldots\left[c_{m}:= \left(\left(\mathcal{L}^c_{\left(i,m\right)}\right)^T\Theta\left(v\right)\right)\right]
			\end{align*}
			Thus, $y\in I^c\Theta\left(v\right)$, underlain by $\left(A_k\delta_i\right)^T\in I^c$. In light of this discovery, we  start a new round of case analysis where we regard the same $y$ as a member of $I^c\Theta\left(v\right)$. Note that whenever this branch, i.e. \ref{enum item: veri casse decide hp sub case nonempty}, is reached, we take a new round of case analysis, thus entering into iteration if we repeatedly reach this branch. Since in each new round, we work with a smaller index, i.e. from $w=\left[v,i\right]$ to $v$,  iteration caused by this branch always terminates, ending in one of the branches \ref{enum item: veri casse decide fact}, \ref{enum item: veri casse decide rule}, \ref{enum item: veri casse decide ch} or \ref{enum item: veri casse decide hp sub case empty} in the final round. In all these possible final circumstances, previous analysis has shown that  $y\in \mathcal{T}\left(I\right)$.
		\end{enumerate}		    
\end{enumerate}
\item Now we have finished with showing that $I\subseteq \mathcal{T}\left(I\right)$, and the proof is complete.	
\end{enumerate}	
\end{proof}

\subsection{Proof of Theorem \protect\ref{them: cut theorem infi model}}

\noindent\textbf{Theorem \protect\ref{them: cut theorem infi model} (Conservative Model Extension). }
\textit{	Let a
	logic program $P_\Sigma$ have coinductive model $\mathcal{M}$, 
	sequent $\Sigma;P\looparrowright H$ have a \textit{co-hohh} proof that only involves guarded atoms.
	Let $H_1,\ldots, H_n$ be distinct term-form ground instances of $H$ involving only guarded atoms, and such that for each $H_k$ ($1\leq k\leq n$ ),   if  
	$A\in\mathit{body}\, H_k$, then  $A^T\in\mathcal{M}$. Let $P\cup\{H_1,\ldots, H_n\}$ have a coinductive model $\mathcal{M}'$.
	Then, $\mathcal{M}=\mathcal{M}'$.}

\begin{proof}
	We first show $\mathcal{M}\subseteq\mathcal{M}'$.
	Let $\mathcal{T}_{P}$ and $\mathcal{T}^{+}_{P}$ denote immediate consequence operators for $P$ and $P\cup\{H_1,\ldots, H_n\}$, respectively. By definition of a coinductive model, we have \[\mathcal{M}=\bigcup\{I\mid I\subseteq\mathcal{T}_P\left(I\right)\}\qquad\mathcal{M}'=\bigcup\{I\mid I\subseteq\mathcal{T}^{+}_{P}\left(I\right)\}\]	
	Note that \[\{I\mid I\subseteq\mathcal{T}_P\left(I\right)\}\subseteq \{I\mid I\subseteq\mathcal{T}^{+}_{P}\left(I\right)\}\]therefore $\mathcal{M}\subseteq\mathcal{M}'$.
	
	Then we show $\mathcal{M}'\subseteq\mathcal{M}$.
	Given a set $I$ that is a post-fixed point of $\mathcal{T}^{+}_{P}$, we distinguish two cases concerning how $I$ is a post-fixed point, as follows. 
	\begin{enumerate}[leftmargin=0pt]
		\item For all $x\in I$, there exist $F\in P$, with tree-form ground instance $F'$, such that $x$ is the head of $F'$, and the body of $F'$ is a subset of $I$. In this case, $I$ is also a post-fixed point of $\mathcal{T}_{P}$.
		\item There exist $x_1,\ldots,x_m\in I$ ($1 \leq m\leq n$, where the inequality $m\leq n$ is explained later), called \emph{outstanding members} of $I$, such that for all $x_i$ ($1\leq i\leq m$)
		\begin{enumerate}
			\item $x_i$ is the tree-atom equivalent to the head of some $H_k$, and, the set of tree-atoms equivalent to the body of $H_k$ is a subset of $I$, and
			\item there does \emph{not} exist $F\in P$, with tree-form ground instance $F'$ such that $x_i$ is the head of $F'$, and the body of $F'$ is a subset of $I$.
		\end{enumerate} Note that the total number of outstanding members in a given post-fixed point of $\mathcal{T}^{+}_{P}$ must be finite, and more precisely, no more than $n$, which is the number of $H$'s instances used to augment the program, to which  outstanding members are associated.
		Because of existence of outstanding members, $I$ is \emph{not} a post-fixed point of $\mathcal{T}_{P}$,  but in below we can show that there exist a post-fixed point $I'$ of $\mathcal{T}_{P}$, such that $I\subseteq I'$. As the theorem   assumes that the set of tree-atoms equivalent  to the  body of each $H_k$ is a subset of $\mathcal{M}$,  then by Theorem~\ref{thm: main soundess infinite}, $\{x_1,\ldots,x_m\}\subseteq \mathcal{M}$. Then by Lemma~\ref{lem: coin prf pcpl infi}, there exist a post-fixed point $I_x$ of $\mathcal{T}_{P}$, such that $\{x_1,\ldots,x_m\}\subseteq I_x$. Then let $I'=I \cup I_x$. It is easy to verify that $I'$ is a post-fixed point of $\mathcal{T}_{P}$ and $I\subseteq I'$. 
	\end{enumerate}		
	Now we have shown that a member of $\{I\mid I\subseteq\mathcal{T}^{+}_{P}\left(I\right)\}$ is either a member of $\{I\mid I\subseteq\mathcal{T}_{P}\left(I\right)\}$, or a subset of a member of the latter set. Therefore  $\mathcal{M}'\subseteq\mathcal{M}$. The proof is complete.
\end{proof}

\section{Examples of this Paper in Coq: Horn clauses as Coinductive Data Types}\label{sec:Coq}

In this section, we compare our four examples with alternative encoding in Coq.
Of interest is Coq's reliance on the notion of coinductive data types and checking 
guardedness of the proof terms. Both of these features are absent in the work we present.
Therefore, the coinductive principle is formulated  in this paper without relying on richer supporting infrastructure available in Coq.     

\subsection{Some minimal infrastructure}

\begin{lstlisting}
Variable A : Set.

CoInductive Stream : Set :=
    Cons : A -> Stream -> Stream.
\end{lstlisting}


Notation \lstinline?(* ... *)? below stands for commented out text in Coq.

We could, but do not, use Coq's inductive types \lstinline?Bool? or \lstinline?nat? in these examples, as none of our examples depends on induction on Booleans or natural numbers. 
Below, in order to mimic the minimalistic logic programming style, we define \lstinline?z, nil, O, I, s?  as arbitrary constants in a given set $A$.

\subsection{Example in cofohc}
Firstly, we give a Coq proof for the clause (6) and the goal $member(0, [0| nil])$  from Introduction.
 It uses a dummy function \lstinline?dummy_cons? for list constructor $[\_ | \_]$, to emphasize that no induction on the list structure is needed. It uses Coq's native definition of equality (in \lstinline? x = y?) for brevity.

Note that the clause (6) is defined as coinductive type, with constructor \lstinline? member_cons?.

\begin{lstlisting}
Variable dummy_cons: A -> A -> A.

CoInductive member:  A -> A -> Prop :=
  | member_cons: forall (x : A) (y : A) (t: A),  
	member x (dummy_cons y t) /\   x = y -> member x (dummy_cons y t).

Variable z : A.
Variable nil: A.

Lemma circular_member: member z (dummy_cons z nil).
(* coinductive hypothesis (member z (dummy_cons z nil )) taken: *)
cofix CH.
(* one step resolution with the Horn clause: *)
apply member_cons.
(* split for conjunction, application of coinductive hypothesis: *)
split; trivial; try apply CH.
Qed.

\end{lstlisting}

Coq infers (and checks) that the type \lstinline?member z (dummy_cons z nil)? is inhabited by the proof term
\lstinline? member_cons (conj CH eq_refl)?, and that this term is guarded by the coinductive constructor \lstinline? member_cons?.
Thus type checking together with guardedness gives soundness to the \lstinline?cofix? rule.

Similar construction of proof terms happens in the remaining examples of this section.

We cannot have a similar guardedness checks in coinductive uniform proofs, as we have no proof terms. Thus, the above guardedness check is replaced by rules of Figure~\ref{fig: rules with angles}.

\subsection{Example in cohohc}

For clauses (3)-(5)  in Introduction:

\begin{lstlisting}
Variable O : A.
Variable I : A. 

Inductive bit : A -> Prop :=
| O_cons: bit O
| I_cons: bit I. 

CoInductive bitstream:   Stream -> Prop :=
  | bitstream_cons: forall (x : A) (t: Stream), 
	bitstream t /\   bit (x) ->  bitstream (Cons x  t).

CoFixpoint z_str : Stream :=
Cons O z_str.
\end{lstlisting}

Note the below term \lstinline?z_str? of type \lstinline?Stream? is guarded by constructor \lstinline?Cons?. Definition~\ref{defn: guarded fpt} insures this kind of guardedness in coinductive uniform proofs.

The proof below is then similar to the one we explained above: 

\begin{lstlisting}
Lemma bistreamO: bitstream z_str.
(* coinductive hypothesis (bitstream z_str) taken: *)
cofix CH.
(* forcing one-step stream reduction: *)
rewrite Stream_decomposition_lemma; simpl.
(* one step resolution with the Horn clause: *)
apply bitstream_cons.
(* split for conjunction, application of coinductive hypothesis: *)
split; try apply CH.
apply O_cons.
Qed.



\end{lstlisting}

\subsection{Example in cofohh}

For clause (9)  in Introduction:

\begin{lstlisting}
Variable f: Stream -> Stream.

CoInductive com_bit: A ->  Stream -> Prop :=
  | com_bit_cons: forall (x : A) (t: Stream),  
	com_bit x (f t) /\   bit (x) -> com_bit  x  t.

Lemma com_bit_lem: forall x y, bit x -> com_bit x y.
(* coinductive hypothesis (forall x y, bit x -> com_bit x y) taken: *)
cofix CH.
intros.
(* one step resolution with the Horn clause: *)
apply com_bit_cons.
(* split for conjunction, application of coinductive hypothesis: *)
split; try apply CH; assumption.
Qed.
\end{lstlisting}

Using coinductive lemma in other proofs:

\begin{lstlisting}
Variable t: Stream.

Lemma com_bit_O: com_bit O t.
apply com_bit_lem; apply O_cons.
Qed.

\end{lstlisting}

\subsection{Example in cohohh}

For clause (8) in Introduction:

\begin{lstlisting}

Variable s: A -> A.

CoInductive from:  A -> Stream -> Prop :=
  | from_cons: forall (x : A) (y: Stream),  
	from (s x) y  -> from x (Cons x y).


CoFixpoint fromstr (x : A ) : Stream :=
Cons x (fromstr (s x)).

Lemma fromQ: forall x, from x (fromstr x).
(* coinductive hypothesis (forall x, from x (fromstr x)) taken: *)
cofix CH.
intros.
(* forcing one-step stream reduction: *)
rewrite Stream_decomposition_lemma; simpl.
(* one step resolution with the Horn clause: *)
apply from_cons.
(* application of coinductive hypothesis: *)
apply CH.
Qed.
\end{lstlisting}

Using coinductive lemma in other proofs:

\begin{lstlisting}

Lemma from0: exists z, from O z.
exists (fromstr O).
apply fromQ.
Qed.
\end{lstlisting}
\end{document}